\documentclass[11pt,a4paper,fleqn]{article}
\pdfoutput=1   % required for arXiv submission if pdflatex is used
\usepackage[utf8]{inputenc}
\usepackage{lmodern}
\usepackage{alpha}
\usepackage{amsthm}
\usepackage{parskip}

\hypersetup{%
   pdftitle    = {First results on QCD+QED with C$^\star$ boundary conditions},
   pdfauthor   = {L. Bushnaq, I. Campos, M. Catillo, A. Cotellucci, M. Dale, P. Fritzsch, J. L\"ucke, M. Krstić Marinković, A. Patella, N. Tantalo},
   pdfkeywords = {Lattice QCD and QED, High Performance Computing, C* boundary conditions}
}%

\newlength{\ftsize}
\def\RClogo{\raisebox{-0.50pt}{\includegraphics[height=\ftsize]{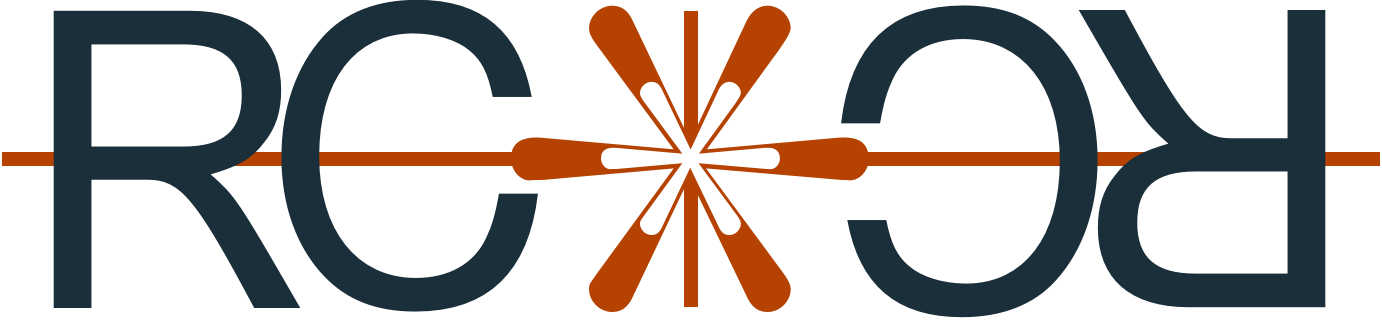}}}

\renewcommand{\vec}[1]{\mathbf{#1}}

\DeclareMathOperator{\tr}{tr}
\DeclareMathOperator{\pf}{pf}
\DeclareMathOperator{\sgn}{sgn}

\newtheorem{proposition}{Proposition}

\begin{document}

% !TEX root = paper.tex

\preprintno{%
HU-EP-22/29-RTG\\
\vfill
}

\title{%
First results on QCD+QED with C$^\star$ boundary conditions
}

\collaboration{\RClogo~collaboration}

\author[trin]{Lucius Bushnaq} 
\author[ifca]{Isabel Campos} 
\author[eth]{Marco Catillo} 
\author[hu]{Alessandro Cotellucci}
\author[TVrome,TVinfn]{Madeleine Dale}
\author[trin]{Patrick Fritzsch}
\author[hu,desy]{Jens L\"ucke}
\author[eth]{Marina Krstić Marinković}
\author[hu,desy]{Agostino Patella}
\author[TVrome,TVinfn]{Nazario Tantalo}

\address[trin]{School of Mathematics, Trinity College Dublin, Dublin 2, Ireland}
\address[ifca]{Instituto de F\'isica de Cantabria \& IFCA-CSIC, Avda. de Los Castros s/n, 39005 Santander, Spain}
\address[eth]{Institut f\"ur Theoretische Physik, ETH Z\"urich, Wolfgang-Pauli-Str. 27, 8093 Z\"urich, Switzerland}
\address[hu]{Humboldt Universit\"at zu Berlin, Institut f\"ur Physik \& IRIS Adlershof, \\Zum Grossen Windkanal 6, 12489 Berlin, Germany}
\address[TVrome]{Universit\`a di Roma Tor Vergata, Dipartimento di Fisica, \\Via della Ricerca Scientifica 1, 00133 Rome, Italy}
\address[TVinfn]{INFN, Sezione di Tor Vergata, Via della Ricerca Scientifica 1, 00133 Rome, Italy}
\address[desy]{DESY, Platanenallee 6, D-15738 Zeuthen, Germany}

\begin{abstract}
Accounting for isospin-breaking corrections is critical for achieving
subpercent precision in lattice computations of hadronic observables. A way to
include QED and strong-isospin-breaking corrections in lattice QCD calculations
is to impose C$^\star$ boundary conditions in space. Here, we demonstrate the
computation of a selection of meson and baryon masses on two QCD and five
QCD+QED gauge ensembles in this setup, which preserves locality, gauge and
translational invariance all through the calculation. The generation of the
gauge ensembles is performed for two volumes, and three different values of the
renormalized fine-structure constant at the U-symmetric point, corresponding to
the SU(3)-symmetric QCD in the two ensembles where the electromagnetic coupling is
turned off. We also present our tuning strategy and, to the extent possible, a
cost analysis of the simulations with C$^\star$ boundary conditions. 

\end{abstract}


\maketitle
% bold math in section headings 
% (AFTER \maketitle, otherwise issue with affiliation counters until fixed)
\makeatletter
\g@addto@macro\bfseries{\boldmath}
\makeatother

\tableofcontents
\pagebreak

%--------------------------------------------------------------
\section{Introduction}\label{sec:intro}
% !TEX root = paper.tex
% !TeX spellcheck = en-US

At the subpercent level of accuracy the hadronic universe is described by
QCD$+$QED. So-called gold-plated hadronic observables, such as meson masses,
leptonic and semileptonic decay rates of light pseudoscalar mesons and leading
hadronic corrections to the muon $g-2$, are calculated in QCD by means of
lattice simulations with a subpercent error (see~\cite{Aoki:2021kgd} for a
recent review). In order to push the frontier of the precision tests of the
hadronic sector of the Standard Model to the subpercent level of accuracy, it is
necessary to perform first-principles lattice QCD$+$QED calculations. By now,
the necessity of including QED and strong-isospin-breaking corrections in
nonperturbative lattice QCD calculations is fully recognized by the lattice
community. In practice, this is done by following different approaches and
strategies, whose review is beyond the scope of this
paper.\footnote{See~\cite{Patella:2017fgk} for a critical, yet slightly
outdated, review of this subject. Besides the references
in~\cite{Patella:2017fgk}, we acknowledge later progress in the
massive-photon-QED approach~\cite{Clark:2022wjy}, and in the infinite-volume-QED
approach~\cite{Blum:2017cer,Feng:2018qpx,Christ:2020hwe}.}
In our long-term research program, aiming at a precise calculation of the
aforementioned gold-plated quantities, we use an approach fully consistent with
the basic principles of quantum field theory. QCD$+$QED is defined in finite
volume by imposing C${}^\star$ boundary conditions along the spatial directions,
which allows to preserve locality, gauge invariance and translational invariance
at all stages of the calculation~\cite{Lucini:2015hfa}. We also choose to
simulate QCD+QED nonperturbatively, even though our setup can be consistently
used also with a perturbative expansion in the electromagnetic coupling
constant \textit{\`a la} RM123~\cite{deDivitiis:2013xla}.

In this paper we present a status update of our project. The main results are
the masses of the $K^{\pm}$, $K^0$, $\pi^\pm$, $D^{\pm}$, $D^0$ and $D_s^\pm$
mesons, the $\Omega^-$ baryon and the octet baryons, calculated on seven
gauge ensembles. The simulations were performed using the \texttt{openQ*D}
code~\cite{openQxD-csic} at fixed value of the bare strong coupling, three
different values of the renormalized fine-structure constant ($\alpha_R \simeq
0, 1/137, 0.04$), two different volumes ($L \simeq 1.6 \text{ fm}, 2.4 \text{
fm}$), and heavier-than-physical light-quark masses corresponding to the
$U$-spin symmetric point ($m_s=m_d$).

In full QCD+QED simulations one expects that, except for particularly
well-behaved observables, the effect of isospin-breaking corrections may be too
small to be resolved against the statistical error. Following
\cite{Borsanyi:2014jba, Horsley:2015eaa, Horsley:2015vla, CSSM:2019jmq}, we
simulate at several values of the fine-structure constant $\alpha_R$. In the
long run, we want to reduce the error on observables by fitting the functional
dependence on $\alpha_R$. We define a renormalization (or matching) scheme which
allows us to compare QCD+QED at different values of the fine-structure constant.
In practice our matching procedure requires the tuning of the bare quark masses
in such a way that certain hadronic observables get a prescribed value. Our
tuning strategy, designed to keep the cost under control, is described in detail
in this paper. 

One peculiarity of C${}^\star$ boundary conditions is that the integration of a
dynamical quark field yields the pfaffian of a matrix which is trivially related
to the Dirac operator. The fermionic pfaffian is real but not necessarily
positive. The absolute value of the pfaffian is included in the simulated
probability distribution of the gauge configurations, while its sign is included
in the observables as a reweighting factor. It is important to stress that this
sign problem is mild, in the sense that the probability to find a negative sign
goes to zero in the continuum limit. Nevertheless, a fully local approach to
QCD$+$QED requires the evaluation of the sign of the fermionic pfaffian, which
we have done systematically for all our configurations. The algorithm used to
calculate the sign is based on standard methods which follow the flow of
eigenvalues of the hermitian Dirac operator with the quark mass. However, it is
worth noticing that our variant is significantly less expensive than similar
algorithms commonly used in the lattice community (e.g.~\cite{Mohler:2020txx}).
A first version of this algorithm was already presented in~\cite{RC:2021tah}, an
even more optimized version is presented in this paper.

We present also a rather detailed investigation of the computational cost for
the generation of our QCD$+$QED configurations with C${}^\star$ boundary
conditions. This is an important piece of information for practitioners who want
to choose the most efficient method to calculate observables in QCD+QED.
However, one should use this piece of information with a grain of salt. For
instance, the generation of QCD+QED configurations is more expensive than QCD
ones, but a fair comparison e.g. between full simulations and RM123 method
should take into account that observables in the latter method (including all
disconnected contributions) are generally more expensive. Similarly, simulations
with C${}^\star$ boundary conditions tend to be more expensive than simulations
with periodic boundary conditions, but a fair comparison e.g. between QCD+QED
with C${}^\star$ boundary conditions and QCD+QED${}_\mathrm{L}$
\cite{Hayakawa:2008an} should take into account that the former has generally
smaller non-universal finite-volume effects than the latter. A fair comparison
between QCD+QED with C${}^\star$ boundary conditions and QCD and infinite-volume
QED should take into account that the latter has much smaller finite-volume
effects than the former, but also a more complicated behavior towards the
continuum limit.\footnote{This is a common feature of integrals of QCD $n$-point
functions with infinite-volume kernels, see e.g. discussion on enhanced lattice
artifacts in \cite{Ce:2021xgd} and Sommer's talk at
Lattice22~\cite{sommer:lattice22}.} Besides efficiency considerations, from the
theoretical point of view it is extremely important to prove that hadron masses
can be computed in a theory with long-range interactions in full compliance with
the basic principles of quantum field theory, whatever it takes in terms of
computational cost.

The paper is organized as follows. The hadronic renormalization scheme used in
this study is presented in section~\ref{sec:lines}, together with the target
values of the parameters in the continuum theory. Section~\ref{sec:overview}
provides an overview of numerical setup and results, and is organized in various
subsections. The used lattice action, the generated gauge ensembles with
corresponding parameters and diagnostic observables are presented in
subsections~\ref{sec:overview:action} and~\ref{sec:overview:ensembles}; meson
and baryon masses and mass-differences, calculated on the original and
mass-reweighted ensembles, are presented in
subsection~\ref{sec:overview:masses}; finite-volume effects are discussed in
subsection~\ref{sec:overview:fv}; a cost analysis for the generation of our
ensembles is presented in subsection~\ref{sec:overview:cost}. Technical details,
which are not essential to the presentation of the general results but are
important to guarantee reproducibility, are postponed to
section~\ref{sec:details} and its subsections. The definition of the used
gradient-flow observables is provided in subsection~\ref{sec:details:flow}; the
definition of the interpolating operators for mesons and baryons, the strategy
used to extract effective masses and identify mass plateaux are discussed in
subsections~\ref{sec:details:meson} and~\ref{sec:details:baryon}; a detailed
presentation of our tuning strategy is given in
subsection~\ref{sec:details:tuning}; the methods used for the statistical
analysis and in particular for the calculation of the autocorrelation times are
outlined in subsection~\ref{sec:details:analysis}; the algorithm used to
calculate the sign of the pfaffian is discussed in
subsection~\ref{sec:details:sign}; the algorithmic parameters are presented in
subsection~\ref{sec:details:algo}.

%--------------------------------------------------------------

\section{Parametrization of four-flavour QCD+QED}\label{sec:lines}
% !TEX root = paper.tex

Continuum four-flavour QCD+QED is a class of theories uniquely defined by six
parameters.\footnote{%
Strictly speaking, the continuum limit of QCD+QED does not exist because of the
triviality of QED. Nevertheless, the continuum limit exists and is universal at
every fixed order in the fine-structure constant in the perturbative regime of
QED, which is the relevant one at typical hadronic energies.
} The particular choice of these parameters is largely arbitrary, and different
choices are often referred to as different \textit{renormalization schemes} or
simply as different \textit{schemes}. For instance, a renormalization scheme
which makes sense in perturbation theory is defined by the
$\Lambda_{\text{QCD}}$ in the $\overline{\text{MS}}$ scheme, the renormalized
fine-structure constant at zero energy (which is scheme independent) and the
four renormalization-group-invariant quark masses. Here we use a nonperturbative
scheme defined by the standard gradient-flow scale $(8t_0)^{1/2}$, the
gradient-flow fine-structure constant $\alpha_R$ at energy $t_0$, and the
following dimensionless observables
\begin{gather}
   \phi_0 = 8 t_0 ( M^2_{K^\pm} - M^2_{\pi^\pm} )
   \ , \nonumber \\
   \phi_1 = 8 t_0 ( M^2_{\pi^\pm} + M^2_{K^\pm} + M^2_{K^0} )
   \ , \nonumber \\
   \phi_2 = 8 t_0 ( M^2_{K^0} - M^2_{K^\pm} ) \alpha_R^{-1}
   \ , \nonumber \\
   \phi_3 = \sqrt{8 t_0} ( M_{D_s^\pm} + M_{D^0}+ M_{D^\pm} )
   \ .
   \label{eq:phi-definitions}
\end{gather}
In the above formulae $M_X$ is the mass of the meson $X$. Details on the
definition of these observables will be given in the subsequent sections. The
convenience of this scheme relies on the fact that all involved observables can
be calculated with very good precision and accuracy on the lattice. However,
this scheme cannot be used directly to find the \textit{physical point}, i.e.
the point in parameter space which describes the real hadronic universe at the
level of subpercent precision, since the scale $t_0$ can not be obtained from
experimental data.\footnote{%
In fact, our $\alpha_R$ cannot be obtained from experimental data either.
However, this is less relevant here because we are truly interested in matching
to the real hadronic universe only up to errors of order $\alpha_R^2$. At this
level of precision $\alpha_R$ is scheme-independent and can be matched to the
PDG value.
} Eventually, the \textit{matching} to the real hadronic universe needs to be
done by going to a hadronic scheme, e.g. by replacing $t_0$ with the mass of the
$\Omega^-$ baryon.

Even though at this stage we are not able to locate the physical point at the
subpercent precision level, we can get close to it by using the value of $t_0$
calculated by various collaborations via QCD simulations. For instance, using
the central value of the CLS determination~\cite{Bruno:2016plf} of
$(8t_0)^{1/2}$, the meson masses and the fine-structure constant from the
PDG~\cite{ParticleDataGroup:2020ssz}, we obtain for the physical point:
\begin{gather}
   \label{eq:physical-point}
   (8t_0^\text{phys})^{1/2} \simeq 0.415 \text{ fm}
   \ , \quad
   \alpha_R^\text{phys} \simeq 1/137 \simeq 0.007299
   \ , \\ \nonumber
   \phi_0^\text{phys} \simeq 0.992
   \ , \quad
   \phi_1^\text{phys} \simeq 2.26
   \ , \quad
   \phi_2^\text{phys} \simeq 2.36
   \ , \quad
   \phi_3^\text{phys} \simeq 12.0
   \ .
\end{gather}
As a side remark, the $\phi$ observables have been designed to be maximally
sensitive to certain combinations of the quark masses. In fact, at the leading
order in SU(3) chiral perturbation theory coupled to QED, one easily shows
\cite{Neufeld:1995mu,Bar:2013ora} that
\begin{subequations}
\begin{gather}
   \phi_0 = A (m_{s,R} - m_{d,R})
   \ , \\
   \phi_1 = 2 A (m_{u,R} + m_{d,R} + m_{s,R}) + 2 B \alpha_R
   \ , \\
   \phi_2 = A \alpha_R^{-1} (m_{d,R} - m_{u,R}) - B
   \ ,
\end{gather}
\end{subequations}
where $A$ and $B$ are some low-energy constants, and $m_{\star,R}$ are the
renormalized quark masses. The observable $\phi_0$ is proportional to the
strange/down mass difference. The observable $\phi_1$ has been already used in
other contexts, e.g.~\cite{Bietenholz:2010jr, Bruno:2014jqa}, and it fixes the
average of the light-quark masses as long as $\alpha_R$ is constant. The
$\phi_2$ observable fixes the ratio between strong and electromagnetic
isospin-breaking effects. Finally, the observable $\phi_3$ is used essentially
to fix the charm quark mass, and has been already used e.g. in
\cite{Hollwieser:2020qri}.

In this paper we present simulations far away from the physical point. There are
two main reasons to consider unphysical values of the parameters.
\begin{enumerate}
   
   \item We expect that at the physical value of $\alpha_R$ we will not be able
   to have enough statistical precision (except for a handful of observables) to
   resolve isospin-breaking effects. Following \cite{Borsanyi:2014jba,
   Horsley:2015eaa, Horsley:2015vla, CSSM:2019jmq}, we simulate at several
   values of the fine-structure constant $\alpha_R$, including $\alpha_R=0$.
   Isospin-breaking effects at the physical value of $\alpha_R$ can be extracted
   by interpolation, while keeping the statistical error under control. 
   
   \item We simulate up and down quarks that are heavier than the physical ones
   in order to make simulations less expensive (while the strange quark is
   lighter than physical). This makes sense especially considering the
   exploratory character of the presented calculation: our current priority is
   to investigate the stability of the chosen simulation setup and to develop
   tuning strategies. The physical value of the quark masses will be approached
   in future studies.

\end{enumerate}
The simulations presented in this paper are performed close to the unphysical
line defined by fixing
\begin{gather}
   (8t_0)^{1/2} = 0.415 \text{ fm}
   \ , \quad
   \phi_0 = 0 
   \ , \quad
   \phi_1 = 2.11
   \ , \quad
   \phi_2 = 2.36
   \ , \quad
   \phi_3 = 12.1
   \ ,
   \label{eq:Uline}
\end{gather}
while $\alpha_R$ is varied from $0$ to $0.04$. Notice that the condition $\phi_0
= 0$ is equivalent to $M_{K^\pm} = M_{\pi^\pm}$ and $m_d=m_s$. When $\phi_0 =
0$, QCD+QED enjoys an enlarged SU(2) flavour symmetry which rotates down and
strange quarks into each other, often called \textit{U-spin symmetry}. For this
reason, we will refer to the line in parameter space defined above as the
\textit{U-symmetric line}. Since $\phi_2$ is kept constant while $\alpha_R$ is
varied, in the $\alpha_R \to 0$ limit one must have that $M_{K^0} = M_{K^\pm}$
which also implies that $m_u=m_d$. Therefore the point $\alpha_R=0$ on the
U-symmetric line is nothing but SU(3)-symmetric QCD. At this stage,
choosing $\phi_1$ close to its physical value ensures that the three degenerate
light quarks have mass roughly equal to the average of the three physical
light-quark masses.\footnote{%
Notice that the chosen value of $\phi_1$ and $\phi_3$ are close but not equal to
the ones given in eq.~\eqref{eq:physical-point}. This is because we chose to
match $\phi_1$ and $\phi_3$ to the ones calculated on our gauge configurations
with $\alpha_R=0$, rather than to the physical ones. This is not essential since
at this point the physical values given in eq.~\eqref{eq:physical-point} are
affected by an unspecified error that comes from the fact that we do not know
the exact value of $t_0$.
} The target lines are represented in figure~\ref{fig:tuning}, together with the
meson masses calculated on our best-tuned ensembles presented in this paper. An
overview of our gauge ensembles and of the observables needed for the tuning is
given in sections~\ref{sec:overview:ensembles} and~\ref{sec:overview:masses}.
Technical details on our tuning strategy are discussed in
section~\ref{sec:details:tuning}.

\begin{figure}
   \begin{center}
      \includegraphics[width=.48\textwidth]{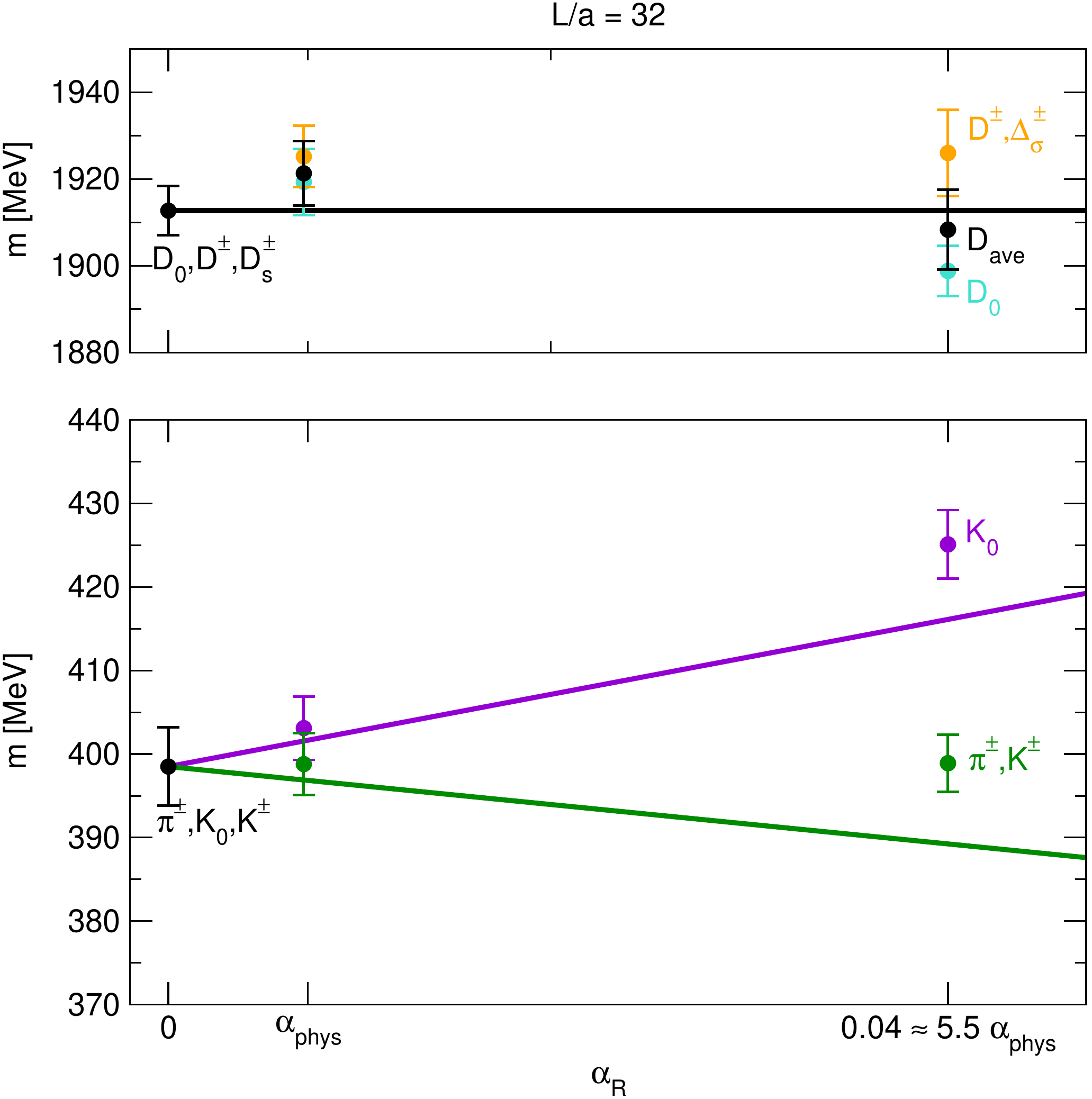}
		\hfill
		\includegraphics[width=.48\textwidth]{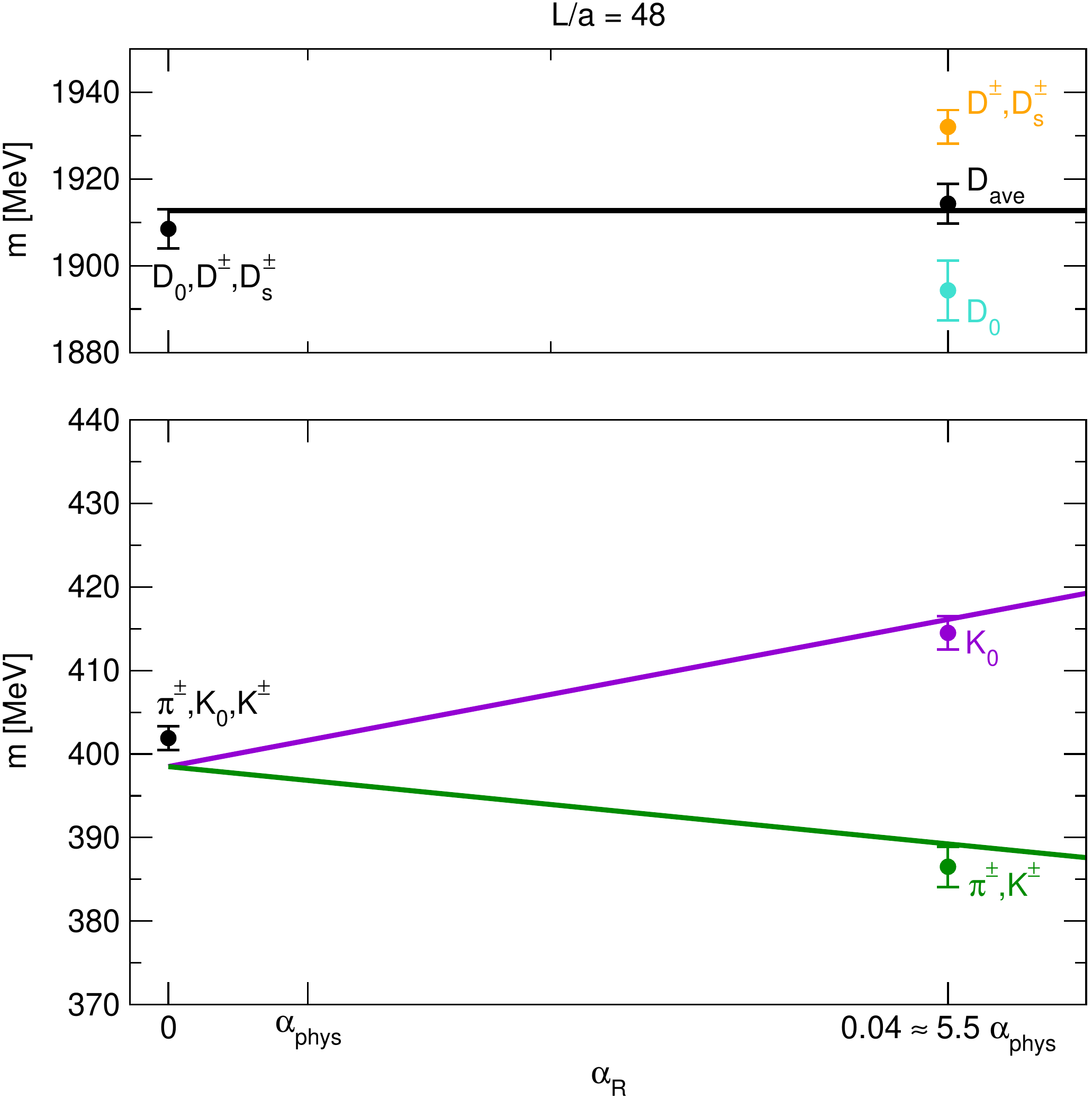}
      \caption{%
		The solid lines are the target values for the $K_0$ mass (purple), the
		$K_\pm$ mass (green), the average $D$ mass $M_{D_{\text{ave}}} =
		\frac{1}{3} ( M_{D_s^\pm} + M_{D_0}+ M_{D^\pm} )$ (black) as functions of
		$\alpha_R$, determined by solving the tuning conditions given in
		eq.~\eqref{eq:Uline}. In the left plots, the data points with error are
		meson masses calculated on the ensemble \texttt{A400a00b324}
		($\alpha_R=0$) and on the reweighted ensembles \texttt{A380a07b324+RW1}
		($\alpha_R \simeq 1/137$) and \texttt{A360a50b324+RW2} ($\alpha_R \simeq
		0.04$). In the right plots, the data points with error are meson masses
		calculated on the ensembles \texttt{B400a00b324} ($\alpha_R=0$) and
		\texttt{C380a50b324} ($\alpha_R \simeq 0.04$). These correspond to our
		best-tuned ensembles. The comparison between left and right plots
		indicates how finite-volume effects affect the tuning. For an overview of
		our gauge ensembles, see section~\ref{sec:overview:ensembles}. The
		calculation of meson masses and $\phi$ observables is presented in
		section~\ref{sec:overview:masses}. Values in MeV are obtained by using the
		reference value $(8 t_0)^{1/2} = 0.415 \text{ fm}$~\cite{Bruno:2016plf}.
      }
      \label{fig:tuning}
   \end{center}
\end{figure}

%--------------------------------------------------------------

\section{Overview of numerical results}\label{sec:overview}

\subsection{Lattice action}\label{sec:overview:action}
% !TEX root = paper.tex

All configurations have been generated with the \texttt{openQ*D}
code~\cite{openQxD-csic}. For a complete description of actions and algorithms
we refer the reader to~\cite{Campos:2019kgw}, while we provide here only a quick
summary.

All our simulations are performed on a $(T/a) \times (L/a)^3$ lattice with
periodic boundary conditions in time, and C$^\star$ boundary conditions in all
spatial directions. We employ the L\"uscher--Weisz discretization for the SU(3)
gauge action, and the Wilson action with an unconventional normalization for the
U(1) gauge action
\begin{gather}
   S_{\mathrm{g,U(1)}}(z) = \frac{1}{8 \pi q_{el}^2 \alpha} \sum_x \sum_{\mu \neq \nu} \left[ 1 - P_{\mu\nu}^{\mathrm{U(1)}}(x) \right]
   \ , \label{eq:U(1)_gauge_action}
\end{gather}
where $P_{\mu\nu}^{\mathrm{U(1)}}(x)$ is the plaquette in $x$ extending in the
directions $\mu$ and $\nu$, constructed with the compact U(1) field $z(x,\mu)$,
and $\alpha$ is the bare fine-structure constant. In the compact formulation
the electric charge is quantized, and it must be an integer multiple of the
parameter $q_{el}$ which can be chosen arbitrarily. In practice, we set $q_{el}
= 1/6$ which allows us to construct gauge-invariant interpolating operators for
charged hadrons as detailed in section 6 of~\cite{Lucini:2015hfa}. We simulate
at a fixed value $\beta = 3.24$ which corresponds roughly to a lattice spacing
of $a \simeq 0.054 \text{ fm}$, and several values of the bare fine-structure
constant $\alpha$.

We simulate four flavours of $O(a)$-improved Wilson fermions. In particular, we
consider the non-physical case in which the up and down quarks are heavier than
physical, the strange quark is lighter than physical, and the down and strange
quarks are degenerate. In the case of QCD+QED the improved Wilson--Dirac
operator includes two Sheikholeslami--Wohlert (SW) terms: the first one depends
on the SU(3) field tensor with coefficient $c_{\mathrm{sw}}^{\mathrm{SU(3)}}$,
and the second one depends on the U(1) field tensor with coefficient
$c_{\mathrm{sw}}^{\mathrm{U(1)}}$. For our QCD ensembles we use the improvement
coefficient $c_{\mathrm{sw}}^{\mathrm{SU(3)}}$, non-perturbatively determined in
\cite{Fritzsch:2018kjg}. For our QCD+QED ensembles we use the same value of
$c_{\mathrm{sw}}^{\mathrm{SU(3)}}$ which is correct up to $O(\alpha)$ terms, and
$c_{\mathrm{sw}}^{\mathrm{U(1)}}=1$ which corresponds to tree-level improvement.

Like in the case of periodic boundary conditions, individual flavours of Wilson
fermions introduce a mild sign problem. After integrating out the fermions, the
path-integral weight is real but generally non-positive. However, the probability
to find a negative sign vanishes in the continuum limit. In the case of
C$^\star$ boundary conditions the sign of the path-integral weight is determined
by the sign of the fermionic pfaffian, and has been systematically calculated on
our gauge ensembles. Details are given in section~\ref{sec:details:sign}.

\subsection{Gauge ensembles}\label{sec:overview:ensembles}
% !TEX root = paper.tex

We have generated 7 ensembles with three different values of $\alpha$. The
ensembles are named with a string that contains a letter in one-to-one
correspondence with the lattice size ($\texttt{A} = 64 \times 32^3$, $\texttt{B}
= 80 \times 48^3$, $\texttt{C} = 96 \times 48^3$), the approximative mass of the
charged pion, the letter $\texttt{a}$ followed by the two digits in
$\alpha=0.0\texttt{xx}\dots$ denoted by $\texttt{x}$, the letter $\texttt{b}$
followed by the value of $100 \times \beta$. The action parameters for all
ensembles are summarized in table~\ref{tab:parameters}, while the number of
generated configurations and a number of diagnostic observables are summarized
in table~\ref{tab:ncnfgs}.

We observe that, among all observables that we have considered, $t_0/a^2$ has
always the largest integrated autocorrelation time. On our $64 \times 32^3$
lattices this turns out to be about 100 MDU. On the larger lattices the
integrated autocorrelation time of $t_0/a^2$ seems to be smaller; however, it is
reasonable to think that we are just underestimating it because of the reduced
statistics. We have also monitored the topological charge, and we observe that
we do not incur topological freezing despite using periodic boundary
conditions in time (see figure~\ref{fig:topo}).

\begin{figure}
   \begin{center}
      \includegraphics[width=.7\textwidth]{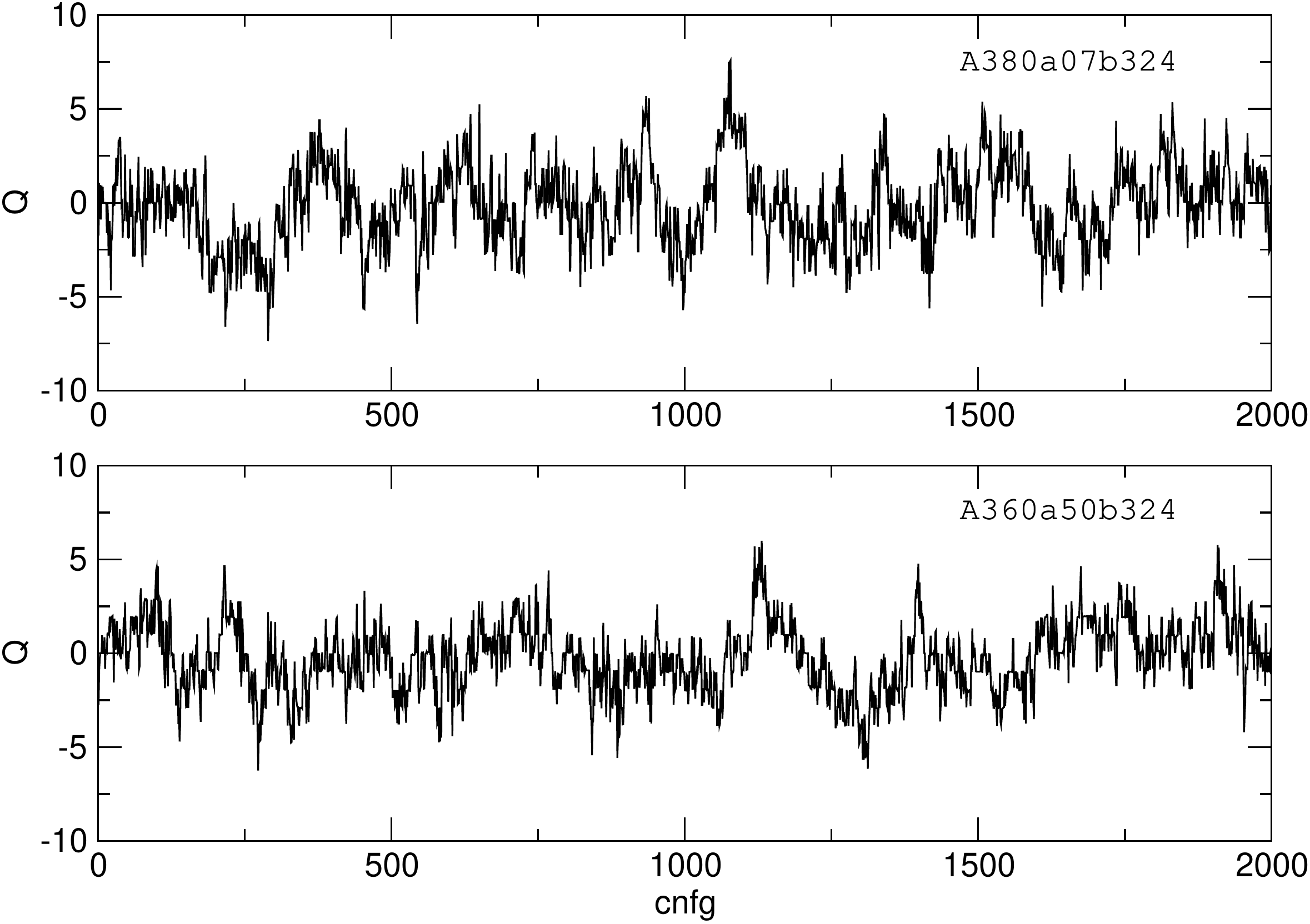}
      \caption{%
      History of the topological charge for two ensembles.
      }
      \label{fig:topo}
   \end{center}
\end{figure}

Compact QED displays a first-order phase transition in bare parameter space
\cite{Azcoiti:1991ng,Campos:1998jp,Arnold:2002jk} which separates a
strong-coupling confining phase and a weak-coupling Coulomb phase. In the pure
gauge theory, the average U(1) plaquette $P^{\mathrm{U(1)}}$ shows a jump across
the phase transition: $P^{\mathrm{U(1)}}$ is small in the confining phase and
close to 1 in the Coulomb phase. Standard weak- and strong-coupling analysis
suggests that the two regimes survive also in presence of fermions. Since we use
a compact action for QED and larger than physical values of $\alpha$, a
legitimate question is whether we are in the Coulomb phase and far enough from
the phase transition. In the conventions of~\cite{Arnold:2002jk}, our largest
value of $\alpha=0.05$ corresponds to
\begin{gather}
   \beta^{\mathrm{U(1)}} = \frac{1}{4\pi q_{el}^2 \alpha} \simeq 57
\end{gather}
which is certainly much larger than the critical value $\beta^{\mathrm{U(1)}}_c
\simeq 1.01$ of the pure gauge theory. The deviation of the average U(1)
plaquette from one is $1 - P^{\mathrm{U(1)}} = 4.19405(21) \times 10^{-3}$ on
the ensemble \texttt{A360a50b324}, which is a clear indication that our
ensembles are always deep in the weak electromagnetic coupling phase.

\begin{table}
   \small
   \begin{center}
      \begin{tabular}{ccccccccc}
         \hline
         ensemble & lattice &
         $\beta$ & $\alpha$ &
         $\kappa_u$ & $\kappa_d=\kappa_s$ & $\kappa_c$ \\
         \hline
         \hline
         \texttt{A400a00b324} & $64 \times 32^3$ & 3.24 & 0 & 0.13440733 & 0.13440733 & 0.12784 \\
         \texttt{B400a00b324} & $80 \times 48^3$ & 3.24 & 0 & 0.13440733 & 0.13440733 & 0.12784 \\
         \hline
         \texttt{A450a07b324} & $64 \times 32^3$ & 3.24 & 0.007299 & 0.13454999 & 0.13441323 & 0.12798662 \\
         \texttt{A380a07b324} & $64 \times 32^3$ & 3.24 & 0.007299 & 0.13459164 & 0.13444333 & 0.12806355 \\
         \hline
         \texttt{A500a50b324} & $64 \times 32^3$ & 3.24 & 0.05 & 0.135479 & 0.134524 & 0.12965 \\
         \texttt{A360a50b324} & $64 \times 32^3$ & 3.24 & 0.05 & 0.135560 & 0.134617 & 0.129583 \\
         \texttt{C380a50b324} & $96 \times 48^3$ & 3.24 & 0.05 & 0.1355368 & 0.134596 & 0.12959326 \\
         \hline
      \end{tabular}
   \end{center}
   \caption{
   Action parameters. All ensembles have C$^\star$ boundary conditions in space
   and periodic boundary conditions in time. The improvement coefficients are
   $c_{\mathrm{sw}}^{\mathrm{SU(3)}} = 2.18859$ and
   $c_{\mathrm{sw}}^{\mathrm{U(1)}} = 1$.
   }\label{tab:parameters}
\end{table}

\begin{table}
   \small
   \begin{center}
      \begin{tabular}{ccccccc}
         \hline
         ensemble & n. cnfg & acc. rate & $\langle e^{- \Delta H} \rangle$ & $\tau_{\mathrm{int}}(t_0)$ & $\tau_{\mathrm{int}}(Q^2)$ & $\tau_{\mathrm{int}}(\alpha_R)$ \\
         \hline
         \hline
         \texttt{A400a00b324} & 2000 & 95\% & 0.9979(55) & 51(18)    & 6.4(2.3)  & --- \\ 
         \texttt{B400a00b324} & 1082 & 98\% & 0.9950(25) & 31(10)    & 8.0(2.8)  & --- \\ 
         \hline
         \texttt{A450a07b324} & 1000 & 94\% & 0.9978(46) & 44(19)    & 6.5(3.0)  & 2.3(1.6) \\ 
         \texttt{A380a07b324} & 2000 & 92\% & 1.0017(46) & 46(15)    & 10.3(3.5) & 2.7(1.5) \\ 
         \hline
         \texttt{A500a50b324} & 1993 & 97\% & 0.9961(21) & 21.4(5.5) & 11.6(2.6) & 1.40(55) \\ 
         \texttt{A360a50b324} & 2001 & 95\% & 0.9956(45) & 47(16)    & 8.5(2.6)  & 1.1(1.0) \\ 
         \texttt{C380a50b324} &  600 & 98\% & 1.004(12)  & 12.5(3.9) & 10.6(4.1) & 3.0(1.2) \\ 
         \hline
      \end{tabular}
   \end{center}
   \caption{
   For each ensemble: the number of configurations which corresponds to the
   number of MD trajectories, the acceptance rate, the diagnostic observable
   $\langle e^{- \Delta H} \rangle = 1$, the integrated autocorrelation times
   (in units of MD trajectories) for the scale $t_0/a^2$, the squared
   topological charge $Q^2$, and the renormalized fine-structure constant
   $\alpha_R$. One MD trajectory is equal to $\tau=2$ MD units.
   }\label{tab:ncnfgs}
\end{table}

\subsection{Tuning and hadron masses}\label{sec:overview:masses}
% !TEX root = paper.tex

Once the theory is discretized on the lattice, it depends on six dimensionless
bare parameters: the inverse bare strong coupling $\beta$, the bare
fine-structure constant $\alpha$ and the hopping parameters $\kappa_f$. In our
simulation strategy these six parameters are treated in different ways: we
\textit{choose} the values of $\beta$ and $\alpha$, while we \textit{tune} the
values of $\kappa_{f=u,d,s,c}$. Changing $\beta$ changes the lattice spacing,
and eventually we want to simulate at different values of $\beta$ and then take
the continuum limit by extrapolating to $\beta \to \infty$. Pretty much like in
QCD, we do care about choosing $\beta$ in a range such that the lattice spacing
is as fine as we can afford, but we do not care about tuning the lattice spacing
to specific values. If we wanted to match the renormalized fine-structure
constant $\alpha_R$ to its physical value, we would need to tune the bare
fine-structure constant $\alpha$. Instead we only want to scan several values of
$\alpha_R$, which means that we can choose the value of $\alpha$ in a reasonable
range and then simply calculate $\alpha_R$. On the other hand, for each chosen
value of $\beta$ and $\alpha$, we want to tune the four hopping parameters
$\kappa_f$ in such a way that the four dimensionless $\phi$ observables match
the chosen values in eq.~\eqref{eq:Uline}. In fact, on the U-symmetric line,
U-spin symmetry fixes $m_d = m_s$ and we truly have to tune only three
parameters. Finally, values of observables in lattice units are converted into
physical units by using the reference value for $t_0$ given in
eq.~\eqref{eq:Uline}.

The tuning of the hopping parameters has been carried out with a combination of
techniques, described in some detail in section~\ref{sec:details:tuning}. In
particular, we have used mass reweighting to explore the space of hopping
parameters in the vicinity of the simulated point. The specific implementation
of the mass reweighting procedure used in this work is described
in~\cite{reweighting-paper}, the peculiarity being that we need to reweight the
determinant of the rational approximation of a generic power of $D^\dag D$. One
can use mass reweighting also to correct for small mistunings, which we have done
to some extent. We label all mass reweighting factors used in this work by the
code $\texttt{RWi}$ where $\texttt{i}$ is an index. In
table~\ref{tab:reweighting} we summarize all mass reweighting factors, together
with the ensembles on which they are calculated and the target quark hopping
parameters.

The calculated values of the lattice spacing $a$ and renormalized fine-structure
constant, with and without mass reweightings, are presented in
table~\ref{tab:a_and_alpha}. Some details on the calculation of these
observables are given in sections~\ref{sec:details:flow}. We have calculated the
masses of the $\pi^\pm$, $K_0$, $K^\pm$, $D_0$, $D^\pm$, $D_s^\pm$ mesons, and
the mass differences for the charged-neutral $K$ mesons and $D$ mesons: results
are presented in table~\ref{tab:masses} and the methods used are described in
section~\ref{sec:details:meson}. It is interesting to notice that, at the tuned
points, we are able to distinguish clearly the $K_0/K^\pm$ mass difference from
zero even at the physical value of $\alpha_R$, while the signal is somewhat less
clear for the $D_0/D^\pm$ mass difference. The $\phi$ observables, which are
used for tuning, are presented in table~\ref{tab:phis}. Notice that, at the
tuned points, we were able to determine $\phi_1$ with a relative statistical error of
about 3\%, $\phi_2$ with a relative statistical error in the range 5-10\%, and $\phi_3$
with a relative statistical error of 0.5\%. Unsurprisingly, since $\phi_2$ is
proportional to the isospin-breaking corrections, it is the hardest to get,
especially at smaller values of $\alpha_R$. The meson masses calculated at the
tuned points are plotted in figure~\ref{fig:tuning}.

We have also calculated the masses of the octet baryons and the $\Omega^-$
baryon in all our small-volume QCD+QED ensembles, together with various baryon
mass differences: the results are presented in tables~\ref{tab:bar_masses}
and~\ref{tab:bar_massdiffs}. A description of the methods used in this
calculation, together with the plots of a selection of effective masses, can be
found in section~\ref{sec:details:baryon}. Here we notice that we obtain baryon
masses with a statistical error in the range 1-5\%, while the statistical error
on the baryon mass differences is less uniform. As expected, the statistical
error on the baryon masses is generally higher for ensembles with heavier pions.
The baryon masses measured on our ensemble \texttt{A360a50b324+RW2} ($\alpha_R
\simeq 0.04$) have significantly larger errors than the other ensembles. This
can be due to a combination of factors: lighter pions, smaller number of
stochastic sources, and perhaps larger effect of the reweighting factor. We plan
to investigate this issue in the future. We have not attempted a systematic
study of excited state contaminations (which is milder for heavier-than-physical
pions) and we do not attempt an estimate of the systematic error due to a
misidentification of the plateau region in our effective masses, but we have
checked the stability of our results against the inclusion of interpolating
operators with different levels of smearing in a generalized eigenvalue problem.

As discussed in~\cite{Lucini:2015hfa}, C$^\star$ boundary conditions partially
break flavour symmetries producing some unphysical mixings, which are pure
finite-volume effects and vanish exponentially fast in the infinite-volume limit
even in QCD+QED. Of the considered baryons, the $\Xi^-$ mixes with the proton,
the $\Xi^0$ mixes with the neutron, and the $\Omega^-$ mixes with the
$\Sigma^{*+}$. In practice, these mixings are generated by quark-disconnected
Wick contractions in the baryon-baryon two-point functions which are allowed
because of the boundary conditions. In our calculation we neglect these
disconnected contributions, which means that we truly consider
partially-quenched baryons made of auxiliary valence quarks for which the mixing
is forbidden. For instance our $\Omega^-$ is truly an $ss's''$ baryon where $s'$
and $s''$ are valence quarks with the same mass and charge as the strange quark
$s$. The two-point functions of these partially-quenched baryons differ from the
two-point functions of the unitary baryons by exponentially suppressed
finite-volume effects.

\begin{table}
   \small
   \begin{center}
      \begin{tabular}{ccccc}
         \hline
         reweighting & ensemble & $\kappa_u$ & $\kappa_d=\kappa_s$ & $\kappa_c$ \\
         \hline
         \hline
         \texttt{RW1} & \texttt{A380a07b324} & 0.13457969 & 0.13443525 & 0.12806355\\ 
         \hline
         \texttt{RW2} & \texttt{A360a50b324} & 0.13553680 & 0.1345960 & 0.12959326 \\ 
         \hline
      \end{tabular}
   \end{center}
   \caption{
   For each mass reweighting factor: the ensemble on which it is calculated and
   the target values of the hopping parameters.
   }\label{tab:reweighting}
\end{table}

\begin{table}
   \small
   \begin{center}
      \begin{tabular}{lcccc}
         \hline
         ensemble(+rw) & $t_0/a^2$ &
         $a$ [fm] & $\alpha_R$ &
         $\pi\sqrt{3}L^{-1}$ [MeV] \\
         \hline
         \hline
         \texttt{A400a00b324}     & 7.402(66) & 0.05393(24)  & 0            & ---        \\
         \texttt{B400a00b324}     & 7.383(40) & 0.05400(14)  & 0            & ---        \\
         \hline           
         \texttt{A450a07b324}     & 7.198(84) & 0.05469(32)  & 0.007076(24) & 613.5(3.6) \\
         \texttt{A380a07b324}     & 7.599(79) & 0.05323(28)  & 0.007081(19) & 630.4(3.3) \\
         \texttt{A380a07b324+RW1} & 7.525(77) & 0.05349(27)  & 0.007080(22) & 627.3(3.2) \\
         \hline           
         \texttt{A500a50b324}     & 7.789(42) & 0.05257(14)  & 0.040772(85) & 638.2(1.7) \\
         \texttt{A360a50b324}     & 8.427(89) & 0.05054(27)  & 0.040633(80) & 663.9(3.5) \\
         \texttt{A360a50b324+RW2} & 8.285(79) & 0.05098(24)  & 0.04069(26)  & 658.2(3.2) \\
         \texttt{C380a50b324}     & 8.400(26) & 0.050625(79) & 0.04073(11)  & 441.86(69) \\
         \hline
      \end{tabular}
   \end{center}
   \caption{
   For each ensemble (possibly with mass reweighting): reference observable
   $t_0/a^2$, lattice spacing $a$ calculated from the measured value of
   $t_0/a^2$, the renormalized fine-structure constant $\alpha_R$, the
   tree-level energy gap of the photon $\pi\sqrt{3}L^{-1}$. Values in physical
   units are obtained by using the reference value $(8 t_0)^{1/2} = 0.415 \text{
   fm}$~\cite{Bruno:2016plf}.
   }
   \label{tab:a_and_alpha}
\end{table}

\begin{table}
   \small
   \begin{center}
		\addtolength{\leftskip} {-1cm}
		\addtolength{\rightskip}{-1cm}
      \begin{tabular}{lcccccc}
         \hline
         ensemble(+rw) &
         $M_{\pi^\pm}=M_{K^\pm}$ & $M_{K_0}$ & $M_{K_0}-M_{K^\pm}$ &
         $M_{D^\pm}=M_{D_s^\pm}$ & $M_{D_0}$ & $M_{D^\pm}-M_{D^0}$ \\
         & [MeV] & [MeV] & [MeV] & [MeV] & [MeV] & [MeV] \\
         \hline
         \hline
         \texttt{A400a00b324}     & 398.5(4.7) & 398.5(4.7) & 0        & 1912.7(5.7) & 1912.7(5.7) & 0         \\
         \texttt{B400a00b324}     & 401.9(1.4) & 401.9(1.4) & 0        & 1908.5(4.5) & 1908.5(4.5) & 0         \\
         \hline                                                                                                
         \texttt{A450a07b324}     & 451.2(4.3) & 451.6(4.7) & 0.8(1.1) & 1919.8(7.3) & 1916.0(8.0) & 3.6(1.2)  \\
         \texttt{A380a07b324}     & 383.6(4.4) & 390.7(3.7) & 7.01(26) & 1926.4(7.8) & 1921.1(7.6) & 5.03(46)  \\
         \texttt{A380a07b324+RW1} & 398.8(3.7) & 403.1(3.8) & 4.26(31) & 1925.2(7.1) & 1919.3(7.6) & 5.8(1.1)  \\
         \hline                                                                                                
         \texttt{A500a50b324}     & 495.0(2.8) & 519.1(2.5) & 24.0(1.0)& 1901.1(4.1) & 1870.1(4.4) & 31.6(1.6) \\
         \texttt{A360a50b324}     & 358.6(3.7) & 388.8(3.5) & 29.5(2.4)& 1937.8(6.8) & 1912.0(7.7) & 26.0(2.8) \\
         \texttt{A360a50b324+RW2} & 398.9(3.4) & 425.1(4.1) & 26.1(1.3)& 1926(10)    & 1898.8(5.8) & 26.9(2.2) \\
         \texttt{C380a50b324}     & 386.5(2.4) & 414.5(2.0) & 26.89(49)& 1932.0(3.9) & 1894.3(6.9) & 34.5(5.6) \\
         \hline
      \end{tabular}
   \end{center}
   \caption{%
	For each ensemble (possibly with mass reweighting): meson masses, and
	charged-neutral meson mass differences. Values in MeV are obtained by using
	the reference value $(8 t_0)^{1/2} = 0.415 \text{ fm}$~\cite{Bruno:2016plf}.
	Notice that some mesons are degenerate because in our simulations $m_d=m_s$.
	}
   \label{tab:masses}
\end{table}

\begin{table}
   \small
   \begin{center}
      \begin{tabular}{llll}
         \hline
         ensemble(+rw) & $\phi_1$ &
         $\phi_2$ & $\phi_3$ \\
         \hline
         \hline
         \texttt{A400a00b324}     & 2.107(50) & ---      & 12.068(36) \\
         \texttt{B400a00b324}     & 2.143(15) & ---      & 12.042(28) \\
         \hline
         \texttt{A450a07b324}     & 2.703(53) & 0.44(60) & 12.097(51) \\
         \texttt{A380a07b324}     & 1.977(37) & 3.39(14) & 12.132(48) \\
         \texttt{A380a07b324+RW1} & 2.126(39) & 2.13(17) & 12.122(47) \\
         \hline
         \texttt{A500a50b324}     & 3.357(37) & 2.60(11) & 11.864(28) \\
         \texttt{A360a50b324}     & 1.806(35) & 2.41(19) & 12.114(41) \\
         \texttt{A360a50b324+RW2} & 2.208(38) & 2.348(97) & 12.040(58) \\
         \texttt{C380a50b324}     & 2.088(22) & 2.350(44) & 12.020(29) \\
         \hline
         \hline
         target                   & 2.11      & 2.36     & 12.1       \\
			\hline
      \end{tabular}
   \end{center}
   \caption{%
	$\phi$ parameters for each ensemble (possibly with mass reweighting),
	together with the target value used to define the lines of constant physics.
	}
   \label{tab:phis}
\end{table}

\begin{table}
	\small
   \begin{center}
      \begin{tabular}{lccccc}
         \hline
         ensemble(+rw) &
         $M_p=M_{\Sigma^{+}}$ & $M_n=M_{\Xi^0}$ & $M_{\Xi^{-}}=M_{\Sigma^{-}}$ & $M_{\Lambda^0}$ & $M_{\Omega^{-}} = M_{\Delta^-}$  \\
			& [MeV] & [MeV] & [MeV] & [MeV] & [MeV] \\
         \hline
         \hline
         \texttt{A450a07b324}     & 1214(14) & 1215(15) & 1216(16) & 1215(15) & 1473(35) \\
         \texttt{A380a07b324}     & 1147(19) & 1151(19) & 1157(18) & 1151(19) & 1458(26) \\
         \texttt{A380a07b324+RW1} & 1164(15) & 1167(13) & 1175(14) & 1167(13) & 1448(20) \\
         \hline
         \texttt{A500a50b324}     & 1280(15) & 1288(13) & 1339(11) & 1296(13) & 1614(23) \\
         \texttt{A360a50b324+RW2} & 1212(20) & 1226(22) & 1268(32) & 1227(24) & 1584(59) \\
         \hline
      \end{tabular}
   \end{center}
	\caption{%
	For each ensemble (possibly with mass reweighting): baryon masses. Values in
	MeV are obtained by using the reference value $(8 t_0)^{1/2} = 0.415 \text{
	fm}$~\cite{Bruno:2016plf}. Notice that some baryons are degenerate because in
	our simulations $m_d=m_s$.
	}
   \label{tab:bar_masses}
\end{table}

\begin{table}
	\small
   \begin{center}
      \begin{tabular}{lccc}
         \hline
         ensemble(+rw) & $M_{n}-M_{p}$ & $M_{\Xi^0}-M_{\Xi^{-}}$ &  $M_{\Sigma^{+}}-M_{\Sigma^-}$ \\
			& [MeV] & [MeV] & [MeV] \\
         \hline
         \hline
         \texttt{A450a07b324}     & -0.89(0.38) & -2.44(0.49) & -1.77(0.89)  \\
         \texttt{A380a07b324}     & 1.80(0.52) & -8.37(0.75)  & -9.96(0.79)  \\
         \texttt{A380a07b324+RW1} & 0.90(0.37) & -5.97(0.63)  & -6.81(0.68)  \\
         \hline                               
         \texttt{A500a50b324}     & 9.2(1.5) & -38.2(2.4) & -46.7(2.7) \\
         \texttt{A360a50b324+RW2} & 10.5(6.0) & -30.2(4.7) & -52(11) \\
         \hline
      \end{tabular}
   \end{center}
	\caption{%
	For each ensemble (possibly with mass reweighting): baryon mass
	differences. Values in MeV are obtained by using the reference value $(8
	t_0)^{1/2} = 0.415 \text{ fm}$~\cite{Bruno:2016plf}.
	}
   \label{tab:bar_massdiffs}
\end{table}

\subsection{Finite-volume effects}\label{sec:overview:fv}
% !TEX root = paper.tex

Most of the presented ensembles correspond to a $64 \times 32^3$ lattice. In
order to estimate finite volume effects, we have generated two larger lattices:
an $80 \times 48^3$ lattice for $\alpha = 0$ and a $96 \times 48^3$ for
$\alpha=0.05$. This allows us to get an idea of the finite-volume effects in
particular on the light meson masses.

In the QCD case, our spatial volumes correspond to $M_\pi L \simeq 3.5$ for the
\texttt{A400a00b324} ensemble and $M_\pi L \simeq 5.2$ for the
\texttt{B400a00b324} ensemble. It is useful to compare our results with the ones
of the ALPHA collaboration~\cite{Hollwieser:2020qri} obtained with periodic
boundary conditions (the action parameters are identical). In
figure~\ref{fig:fv-qcd} we show the pion mass for two volumes and different
boundary conditions. It is interesting to notice that finite volume corrections
tend to increase the mass in the case of periodic boundary conditions, while
they tend to decrease the mass in the case of C$^\star$ boundary conditions.
This behaviour is captured by chiral perturbation theory at leading order (for
the periodic case see e.g.~\cite{Colangelo:2005gd}):
\begin{gather}
   M_{\text{P}}(L) = M + \frac{\xi}{3} \sum_{\vec{n} \in \mathbb{Z}^3 \setminus \{0\}} \frac{2}{L} K_1(nML)
   \ ,
   \\
   M_{\text{C}}(L) = M - \frac{\xi}{3} \sum_{\vec{n} \in \mathbb{Z}^3 \setminus \{0\}} \frac{1 - 3(-1)^{\sum_k n_k}}{nL} K_1(nML)
   \ ,
\end{gather}
where $n = | \vec{n}|$, $\xi = M^2/(4 \pi F)^2$ with $F$ being the pion decay
constant, $K_1$ is a modified Bessel function of the second kind, the subscripts
P and C denote periodic and C$^\star$ boundary conditions, respectively. In
figure~\ref{fig:fv-qcd} we plot also the result of the simultaneous fits with
the two above functions in the parameters $M$ and $\xi$. On the large volumes,
finite volume effects on the pion mass are surely not larger than 1\%. While the
statistical errors on the pion masses on the smaller volumes are fairly large
and a definite interpretation would require higher statistics,
figure~\ref{fig:fv-qcd} suggests that finite volume effects are sizable in this
case.

\begin{figure}
   \begin{center}
      \includegraphics[width=.7\textwidth]{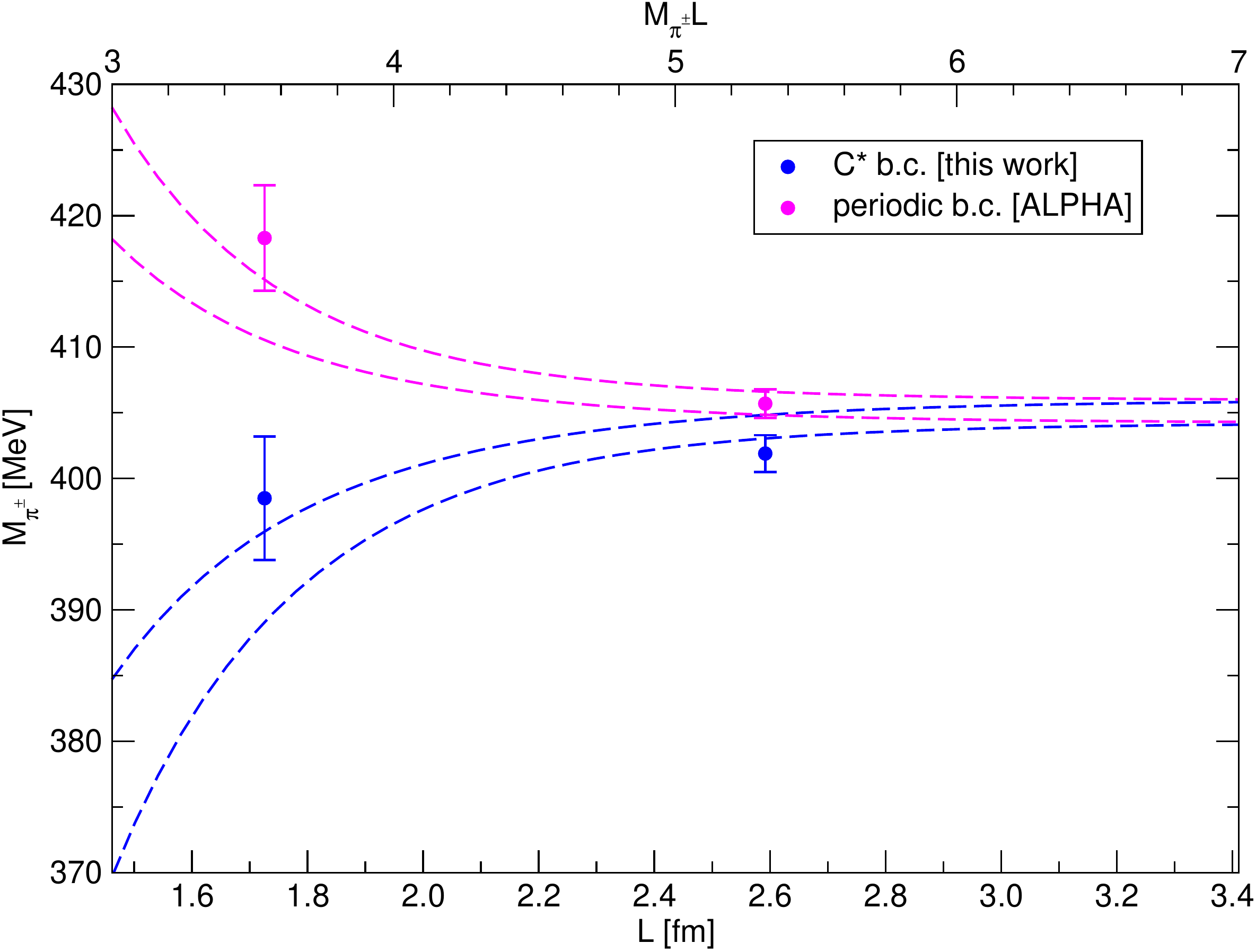}
      \caption{%
      Pion masses for our \texttt{A400a00b324} and \texttt{B400a00b324}
      ensembles with C$^\star$ boundary conditions, and for the \texttt{A1} and
      \texttt{A2} from~\cite{Hollwieser:2020qri} with periodic boundary
      conditions. In both cases physical units have been introduced by setting
      $(8 t_0)^{1/2} = 0.415 \text{ fm}$~\cite{Bruno:2016plf}. The curves are
      the results of the simultaneous fit to the LO $\chi$PT formulae with
      uncertainty bands.
      }
      \label{fig:fv-qcd}
   \end{center}
\end{figure}

When QCD is coupled to QED the pion mass gets power corrections which vanish
like inverse powers of the volume. In the case of C$^\star$ boundary conditions
these finite-volume effects have been derived in~\cite{Lucini:2015hfa}:
\begin{gather}
   M_{\text{C}}(L) = M - \alpha \left\{
   \frac{q^2 \zeta(1)}{2 L} + \frac{q^2 \zeta(2)}{\pi M L^2}
   + \sum_{\ell=0}^\infty \frac{(-1)^\ell \zeta(2\ell+2)}{4\pi M L^{4+2\ell}} T_\ell
   \right\} + O(\alpha^2)
   \ ,
\end{gather}
where $q$ is the charge of the pion, the generalized zeta function $\zeta(p)$ is
defined by
\begin{gather}
   \zeta(p) = \sum_{\vec{n} \in \mathbb{Z}^3 \setminus \{0\}} \frac{(-1)^{\sum_k n_k}}{|\vec{n}|^p}
   \ ,
\end{gather}
and the coefficients $T_\ell$ are related to the coefficients of the Taylor
expansion with respect to the on-shell photon energy of the forward Compton
scattering amplitude of the pion (for more details see~\cite{Lucini:2015hfa}).
The first two terms of the $1/L$-expansion are structure independent and they
have been already subtracted in all hadron masses presented in the tables of
this paper, accordingly to the procedure detailed in
section~\ref{sec:details:meson}. The only exception is table~\ref{tab:fv-data},
in which we present the masses of charged mesons for a couple of ensembles
without the subtraction of the structure-independent finite-volume effects, and
with the subtraction of only the $1/L$ term. For the charged pion on our
\texttt{C380a50b324} ensemble ($L/a=48$ and $\alpha=0.05$), the
structure-independent contributions to the finite-volume effects turn out to be
about 0.9\% and most of the effect comes from the leading $1/L$ term. Hence it
is reasonable to assume that QED finite-volume effects are well under control in
our largest volume, even though a more detailed study would be desirable.

\begin{table}
   \small
   \begin{center}
      \addtolength{\leftskip} {-1cm}
		\addtolength{\rightskip}{-1cm}
      \begin{tabular}{l|ccc|ccc}
         \hline
         & \multicolumn{3}{c|}{$M_{\pi^{\pm}}=M_{K^{\pm}}$ [MeV]} & \multicolumn{3}{c}{$M_{D^{\pm}} = M_{D^{\pm}_s}$ [MeV]} \\
         ensemble(+rw) & no-FV & LO-FV & NLO-FV & no-FV & LO-FV & NLO-FV \\
         \hline
         \hline
         \texttt{A360a50b324+RW2} & 393.4(3.4) & 397.7(3.4) & 398.9(3.4) & 1922(10) & 1926(10) & 1926(10) \\
         \texttt{C380a50b324}     & 383.1(2.4) & 386.0(2.4) & 386.5(2.4) & 1929.0(3.9) & 1931.9(3.9) & 1932.0(3.9) \\
         \hline
      \end{tabular}
   \end{center}
   \caption{
   For the two ensemble with $\alpha_R \simeq 0.04$ and tuned values of the
   quark masses: $\pi^{\pm}$ and $D_{\pi^{\pm}}$ masses calculated with no
   subtraction of the structure-independent finite-volume corrections (no-FV),
   with the subtraction of only the $1/L$ finite-volume correction (LO-FV), with
   the subtraction of the $1/L$ and $1/L^2$ finite-volume corrections (NLO-FV).
   The $1/L^2$ finite volume correction produces a 0.3\% shift on the
   $\pi^{\pm}$ mass on the smaller volume, and a 0.1\% shift on the $\pi^{\pm}$
   mass on the larger volume. Values in MeV are obtained by using the reference
   value $(8 t_0)^{1/2} = 0.415 \text{ fm}$~\cite{Bruno:2016plf}.
   }
   \label{tab:fv-data}
\end{table}

\subsection{Note on simulation cost}\label{sec:overview:cost}
% !TEX root = paper.tex

Our production runs have been performed on a variety of machines: Lise at HRLN,
Marconi at CINECA, Eagle at PSNC and Piz Daint at CSCS. In order to be able to
compare the production cost, we have measured the time needed to generate a
thermalized configuration on Lise\footnote{%
Lise has 1236 standard nodes with 384 GB memory, each of them with 2 CPUs. The
CPUs are Intel Cascade Lake Platinum 9242 (CLX-AP) with 48 cores each. The nodes
are connected with an Omni-Path network with a fat tree topology, 14 TB/s
bisection bandwidth and 1.65 $\mu$s maximum latency. Source:
\url{https://www.hlrn.de/supercomputer-e/hlrn-iv-system/?lang=en}.
}
at HLRN for all gauge ensembles. The results are shown in table~\ref{tab:cost};
in particular, we report the specific cost, i.e. the cost in core$\times$seconds
per molecular dynamics unit divided by the number of lattice points.\footnote{%
The reader familiar with \texttt{openQ*D} knows that C$^\star$ boundary
conditions are implemented by means of an orbifold procedure which effectively
doubles the lattice size. In the code, one distiguishes between physical and
extended (i.e. doubled) lattice. Throughout this paper we always refer to the
physical lattice. When we talk about lattice volume, we always refer to the
volume of the physical lattice. In particular, in order to reconstruct the total
cost from table~\ref{tab:cost}, one needs to multiply the specific cost times
the number of points of the physical lattice.
} Comparing specific costs makes sense particularly if the machine is always
used in a regime of reasonably good scaling, which seems to be the case for the
presented runs. In table~\ref{tab:cost} we have also reported the production
cost for the QCD ensemble \texttt{A1} generated by the ALPHA collaboration on
Lise and presented in~\cite{Hollwieser:2020qri}. The production cost is not
reported in the paper and has been kindly provided by the authors
of~\cite{Hollwieser:2020qri}.

The ensemble \texttt{A1} can be compared directly with our ensemble
\texttt{A400a00b324}. The two QCD ensemble use the same discretization of the
action, the same values of $\beta$, bare quark masses and improvement
coefficients, but differ for the volume and the boundary conditions. The
ensemble \texttt{A1} uses a $96 \times 32^3$ lattice with open boundary
conditions in time and periodic boundary conditions in space, while our ensemble
\texttt{A400a00b324} uses a $64 \times 32^3$ lattice with periodic boundary
conditions in time and C$^\star$ boundary conditions in space. The two ensembles
differ also by a number of algorithmic parameters, rendering a precise cost
comparison between the two ensembles complicated. In the following, we use a
back-of-the-envelope calculation to argue that we understand the most important
effects contributing to the cost ratio:
\begin{gather}
   \left[ \frac{
   \text{specific cost}(\texttt{A400a00b324})
   }{
   \text{specific cost}(\texttt{A1})
   }
   \right]_{\text{measured}}
   \simeq
   \frac{0.42}{0.35}
   =
   1.2
   \ .
   \label{eq:cost-ratio}
\end{gather}

\begin{enumerate}
   
   \item Because of C$^\star$ boundary conditions, the Dirac operator is a
   matrix that acts on a vector space with dimension $24V$, as opposed to $12V$
   in case of periodic boundary conditions. This means that the application of
   the Dirac operator on a single pseudofermion costs twice as much as the
   periodic case if the physical volume is the same. C$^\star$ boundary
   conditions contribute with a factor of two to the
   ratio~\eqref{eq:cost-ratio}.
   
   \item A three-level integrator has been used in both cases: 8 steps of
   second-order Omelyan in the outermost level, 1 step of fourth-order Omelyan
   in the intermediate level, 1 step (for \texttt{A400a00b324}) or 2 steps (for
   \texttt{A1}) of fourth-order Omelyan in the innermost level. The difference
   in the innermost level is expected not to have a significant impact on the
   total cost since only the gauge forces are integrated in that level in both
   cases. Therefore, we estimate that the difference in integrator steps
   contributes with a factor of roughly one to the ratio~\eqref{eq:cost-ratio}.
   
   \item In the case of the ensemble \texttt{A1}, the HMC with frequency
   splitting has been used for the up/down doublet and two different rational
   approximations have been used for the strange and charm. In the case of the
   ensemble \texttt{A400a00b324}, a single rational approximation has been used
   for the degenerate up/down/strange triplet, and a separate rational
   approximation has been used for the charm. In spite of this difference, it
   turns out that the total number of fermionic forces that need to be
   calculated is not so different in the two cases. For \texttt{A400a00b324} and
   \texttt{A1}, the outermost level integrates six and five pseudofermion
   forces, respectively, and the intermediate level integrates six and eight
   pseudofermion forces, respectively. Since the forces of the intermediate
   level are calculated much more often than the ones on the outer level, we
   consider only the intermediate level for this back-of-the-envelope
   calculation. Therefore, we estimate that the difference in pseudofermion
   forces contributes with a factor of roughly $6/8=0.75$ to the
   ratio~\eqref{eq:cost-ratio} (reducing the cost gap between the two
   ensembles).
   
   \item Smaller residues have been typically used in \texttt{A1} for the
   solvers used to calculate pseudofermion forces, reducing the cost gap between
   \texttt{A400a00b324} and \texttt{A1} even further. The impact of this effect
   has been estimated by looking at how many times the Dirac operator is applied
   by the various solvers. We estimate that the difference in residues
   contributes with a factor of roughly $0.84$ to the ratio~\eqref{eq:cost-ratio}
   (further reducing the cost gap between the two ensembles).
   
\end{enumerate}
Multiplying all the above factors together we obtain the following estimate for the cost ratio
\begin{gather}
   \left[ \frac{
   \text{specific cost}(\texttt{A400a00b324})
   }{
   \text{specific cost}(\texttt{A1})
   }
   \right]_{\text{estimated}}
   \simeq
   2 \times 1 \times 0.75 \times 0.84
   =
   1.26
   \ ,
\end{gather}
which is remarkably close to the measured ratio~\eqref{eq:cost-ratio}.

The comparison between \texttt{A380a07b324} and \texttt{A360a50b324} (fairly
similar algorithmic parameters where used in these two runs) suggests that the
computational cost does not depend significanly on $\alpha$ in the interesting
region. The comparison between these two ensembles on the one hand and the
ensemble \texttt{A400a00b324} on the other hand shows a clear cost gap between
the QCD+QED and QCD simulations, yielding e.g.
\begin{gather}
   \left[ \frac{
   \text{specific cost}(\texttt{A360a50b324})
   }{
   \text{specific cost}(\texttt{A400a00b324})
   }
   \right]_{\text{measured}}
   \simeq
   \frac{1.05}{0.42}
   =
   2.5
   \ .
\end{gather}
In a certain measure this is due to physics: the QCD ensemble has an SU(3)
flavour symmetry which allows us to use a single rational approximation for the
three light quarks, while the QCD+QED ensembles have only an SU(2) flavour
symmetry forcing us to use a rational approximation for the up quark and a
different rational approximation for the down/strange quarks. However we also
notice that we need to increase the number of integration steps in our QCD+QED
ensembles (from 8 to 12 in the outermost level) in order to keep the acceptance
rate to a reasonable level. The source of this effect is unclear, and we plan to
investigate in the future whether it is possible to avoid it by optimizing the
algorithmic parameters.

The specific cost is essentially the same for the \texttt{A380a07b324} and
\texttt{A450a07b324} ensembles, modulo fluctuations in performance. In
particular we detect no significant dependence on the light quark masses.

The increase in specific cost from the \texttt{A360a50b324} to the
\texttt{C380a50b324} is completely accounted for by the increase in the number
of integration steps in the outermost level (from 12 to 18) needed to compensate
the reduction in acceptance rate due to the larger volume. However, a posteriori
we have overdone it, and we could have probably used some intermediate value.
The increase in specific cost from the \texttt{A400a00b324} to the
\texttt{B400a00b324} is partly accounted for by the increase in the number of
integration steps in the outermost level (from 12 to 16), while the extra cost
may be due to a decrease in efficiency due to use of a highly-asymmetric local
lattice.

More details on the choice of algorithmic parameters are provided in
section~\ref{sec:details:algo}.

\begin{table}
   \small
   \begin{center}
   \begin{tabular}{cccc}
   \hline
   ensemble & global volume & n. cores
   & specific cost
   \\
   \rule{0mm}{4mm}
   & & 
   & $ \left[ \frac{\text{cores}\times\text{secs}}{\text{MDUs}\times\text{points}} \right]$
   \\[2mm]
   \hline
   \hline
   \texttt{A1} \cite{Hollwieser:2020qri} & $96 \times 32^3$ & 6144
   & 0.35 \\
   \hline
   \hline
   \texttt{A400a00b324} & $64 \times 32^3$ & 4096
   & 0.42\\
   \texttt{B400a00b324} & $80 \times 48^3$ & 2560
   & 0.62 \\
   \hline
   \texttt{A450a07b324} & $64 \times 32^3$ & 4096
   & 1.07 \\
   \texttt{A380a07b324} & $64 \times 32^3$ & 4096
   & 1.03 \\
   \hline
   \texttt{A500a50b324} & $64 \times 32^3$ & 4096
   & 0.88 \\
   \texttt{A360a50b324} & $64 \times 32^3$ & 4096
   & 1.05 \\
   \texttt{C380a50b324} & $96 \times 48^3$ & 3072
   & 1.40 \\
   \hline
   \end{tabular}
   \end{center}
   \caption{%
   Cost comparison of all production runs presented in this paper, plus the
   $N_f=3+1$ QCD ensemble \texttt{A1} produced by the ALPHA
   collaboration~\cite{Hollwieser:2020qri}. All wall times have been measured on
   Lise at HLRN. For each run we report the global lattice volume, the number of
   cores, and the specific cost i.e. the cost in coresecs per molecular dynamics
   unit (MDU) divided by the global volume.
   }\label{tab:cost}
\end{table}

%--------------------------------------------------------------
\section{Technical details}\label{sec:details}

\subsection{Flow observables}\label{sec:details:flow}
% !TEX root = paper.tex

The gradient flow is used to define the auxiliary observable $t_0$ and the
renormalized fine-structure constant $\alpha_R$. In particular, we use the
Wilson-flow discretization~\cite{Luscher:2010iy} for the SU(3) flow equation
\begin{gather}
   a^2 \partial_t U_t(x,\mu) = - g^2 \left\{ \partial_{x,\mu} S_{\mathrm{w,SU(3)}}(U_t) \right\} U_t(x,\mu)
   \ ,
\end{gather}
where $U_t$ is the SU(3) gauge field at positive flow time, and
$S_{\mathrm{w,SU(3)}}(U)$ is the standard SU(3) Wilson action. For the U(1) flow
equation we use the obvious generalization
\begin{gather}
   a^2 \partial_t z_t(x,\mu) = - 4\pi \alpha \left\{ \partial_{x,\mu} S_{\mathrm{g,U(1)}}(z_t) \right\} z_t(x,\mu)
   \ ,
\end{gather}
where $z_t$ is the compact U(1) gauge field at positive flow time,
$S_{\mathrm{g,U(1)}}(z)$ is the action given in
eq.~\eqref{eq:U(1)_gauge_action}. If $\hat{G}_{t,\mu\nu}(x)$ and
$\hat{F}_{t,\mu\nu}(x)$ are, respectively, the clover discretizations of the SU(3)
and U(1) field tensors at positive flow time, we define the clover action
densities as
\begin{gather}
   E_{\mathrm{SU(3)}}(t) = \frac{1}{2} \sum_{\mu\nu} \langle \tr \hat{G}_{t,\mu\nu}^2 \rangle
   \ , \qquad
   E_{\mathrm{U(1)}}(t) = \frac{1}{4 q_{el}^2} \sum_{\mu\nu} \langle \hat{F}_{t,\mu\nu}^2 \rangle
   \ .
\end{gather}
The auxiliary observable $t_0$ is defined as usual by means of the equation
\begin{gather}
   t_0^2 E_{\mathrm{SU(3)}}(t_0) = 0.3
   \ ,
\end{gather}
while the renormalized fine-structure constant is defined at the scale $t_0$ as
\begin{gather}
   \alpha_R = \mathcal{N} t_0^2 E_{\mathrm{U(1)}}(t_0)
   \ .
\end{gather}
Following~\cite{Fritzsch:2013je}, the normalizaton $\mathcal{N}$ is chosen
in such a way that $\alpha_R$ coincides with $\alpha$ at tree level in the
lattice perturbative expansion. Its explicit formula is given by
\begin{gather}
   \mathcal{N}^{-1}
   =
   \frac{2\pi t_0^2}{T L^3} \sum_p \frac{
   \sum_{\mu\nu} \mathring{p}_\mu^2 c_\nu^2 - \sum_\mu \mathring{p}_\mu^2 c_\mu^2
   }{
   \sum_\mu \hat{p}_\mu^2
   }
   e^{-2 t_0 \sum_\mu \hat{p}_\mu^2}
   \ ,
\end{gather}
where the sum runs over all momenta allowed by the boundary conditions
\begin{gather}
   p_0 \in \frac{2\pi a}{T} \left\{ 0 , 1 , 2, \dots, \frac{T}{a}-1 \right\}
   \ , \qquad
   p_k \in \frac{\pi a}{L} \left\{ 1 , 3 , 5, \dots, \frac{2L}{a}-1 \right\}
   \ ,
\end{gather}
and the following definitions have been used:
\begin{gather}
   \hat{p}_\mu = \frac{2}{a} \sin \left( \frac{a p_\mu}{2} \right)
   \ , \qquad
   \mathring{p}_\mu = \frac{1}{a} \sin (a p_\mu)
   \ , \qquad
   c_\mu = \cos \left( \frac{a p_\mu}{2} \right)
   \ .
\end{gather}

\subsection{Meson masses and $\phi$ observables}\label{sec:details:meson}
% !TEX root = paper.tex

Since we use the compact formulation of QED, we do not need to fix the gauge.
With this choice, physical states (even charged ones) are invariant under SU(3)
and U(1) local gauge transformations. In finite volume with C$^\star$ boundary
conditions, global U(1) gauge symmetry is broken down to the
$\mathbb{Z}_2$ subgroup which allows to distinguish states with even and odd
electric charge (see~\cite{Lucini:2015hfa} for an extended discussion).
Gauge-invariant quark bilinears are constructed as usual, but the elementary
quark fields $\psi_f$ and $\bar{\psi}_f$ need to be replaced with the dressed
ones:
\begin{gather}
   \Psi_f(x) = \mathcal{D}_f(x) \psi_f(x)
   \ , \qquad
   \bar{\Psi}_f(x) = \bar{\psi}_f(x) \mathcal{D}_f^*(x)
   \ ,
   \label{eq:dressed-psi}
\end{gather}
where the dressing factor $\mathcal{D}_f(x)$ has been chosen to be the
$(\hat{q}_f/2)$-th power of the spatial U(1) Polyakov loops starting from $x$,
averaged over the three spatial directions, i.e.
\begin{gather}
   \mathcal{D}_f(x) = \frac{1}{3} \sum_{k=1}^3 \prod_{s=0}^{L/a} z^{\hat{q}_f/2}(x+a s \hat{k},k)
   \ .
\end{gather}
One easily checks that the dressed quark fields are invariant under local U(1)
gauge transformations thanks to C$^\star$ boundary conditions. Moreover, the
dressing factor $\mathcal{D}_f(x)$ is invariant under 90$^\circ$ rotations
around $x$. The parameter $\hat{q}_f$ is the charge of the quark field in units
of the gauge-action parameter $q_{el}$. With our choice $q_{el} = 1/6$, up-type
quarks have $\hat{q}=4$ and down-type quarks have $\hat{q}=-2$. In all cases the
quantity $\hat{q}_f/2$, which appears in the exponent of the dressing factor, is
an integer.

Because of the boundary conditions, eigenstates of the momentum operator are
automatically eigenstates of the charge conjugation operator. In particular,
C-even fields are periodic and C-odd fields are antiperiodic in all spatial
directions. In order to construct zero-momentum fields, one needs to construct
C-even combinations first. The C-even zero-momentum interpolating operators of
pseudoscalar mesons are given by
\begin{gather}
   P_{fg}(x_0)
   =
	\sum_{\vec{x}} \left\{
   \bar{\Psi}_{f} \gamma_5 \Psi_{g}(x_0,\vec{x})
   + \bar{\Psi}_{g} \gamma_5 \Psi_{f}(x_0,\vec{x})
	\right\}
   \ ,
\end{gather}
for generic flavour indices $f$ and $g$, and two-point functions are defined as
\begin{gather}
   C_{fg}(x_0) = \langle P_{fg}(x_0) P_{fg}(0) \rangle \ ,
\end{gather}
In this work we consider only the two-point functions with $f \neq g$ which can
be written in terms of standard quark-connected diagrams. We stress that, even
though the interpolating operators $P_{fg}(x)$ are non-local because of the
dressing factors, they are local in time, and the zero-momentum two-point
function has a standard spectral representation which allows the extraction of
Hamiltonian eigenstates from its exponential decay at large $x_0$.

Given the zero-momentum two-point function $C(x_0)$, we define the effective
mass $M(x_0)$ by solving the following equation numerically:
\begin{gather}
    \frac{C(x_0+a)}{C(x_0)}
    =
    \frac{\cosh \left[\left(x_0+a-\frac{T}{2}\right) \, M(x_0) \right]}{\cosh\left[ \left(x_0-\frac{T}{2}\right) \, M(x_0)\right]} \ .
\end{gather}

Hadron masses get power-law finite-volume corrections due to the coupling to the
photon. In the case of C$^\star$ boundary conditions these have been calculated
in~\cite{Lucini:2015hfa}. The LO and NLO corrections in $1/L$ are universal and
are subtracted from the effective mass by means of the formula
\begin{gather}
   M_c(x_0) = M(x_0) - \alpha_R q^2 \left\{ \frac{\zeta(1)}{2L} + \frac{\zeta(2)}{\pi M(x_0) L^2} \right\}
   \ , \label{eq:corrected-effmass}
\end{gather}
where $\zeta(1) = -1.7475645946\dots$ and $\zeta(2) = -2.5193561521\dots$ and
$q$ is the charge of the considered hadron. The meson mass is simply obtained by
fitting the plateaux of the corrected effective mass $M_c(x_0)$ to a constant,
and by checking the stability of the result under variation of the plateau. The
effective $\phi_{1,2,3}(x_0)$ observables have been calculated by applying the
definition~\eqref{eq:phi-definitions} to the corrected effective masses of the
relevant mesons. The $\phi$ observables are obtained by fitting the plateaux of
the corresponding effective quantity to a constant, and by checking the
stability of the result under variation of the plateaux. A selection of
effective masses and $\phi$'s with the corresponding plateau fits are shown in
figures~\ref{fig:effmasses-1} and \ref{fig:effmasses-2}.

\begin{figure}
   \begin{center}
      \begin{minipage}[c]{.625\textwidth}
         \includegraphics[width=\textwidth]{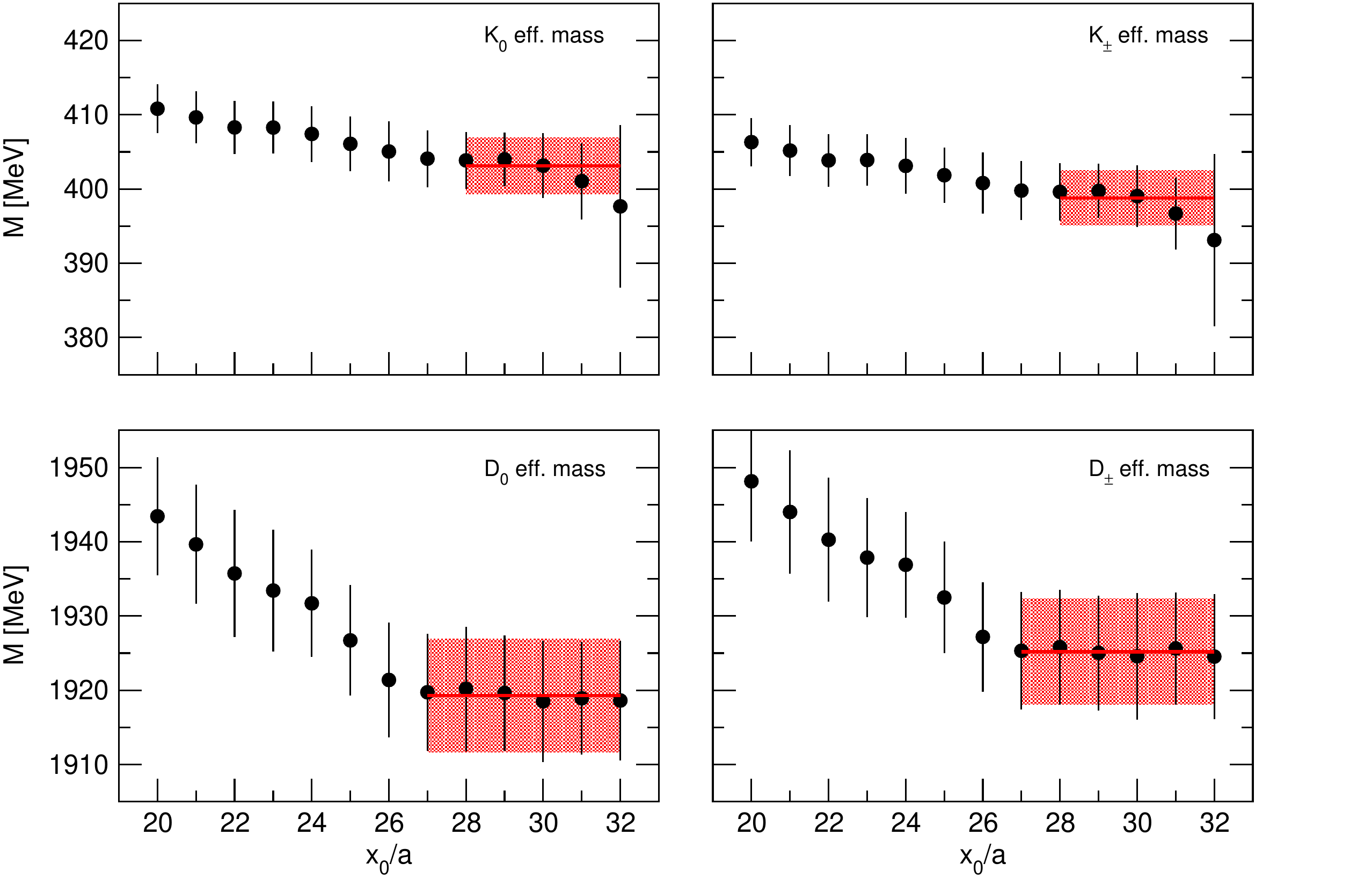}
      \end{minipage}\begin{minipage}[c][6cm][c]{.35\textwidth}
         \includegraphics[width=\textwidth]{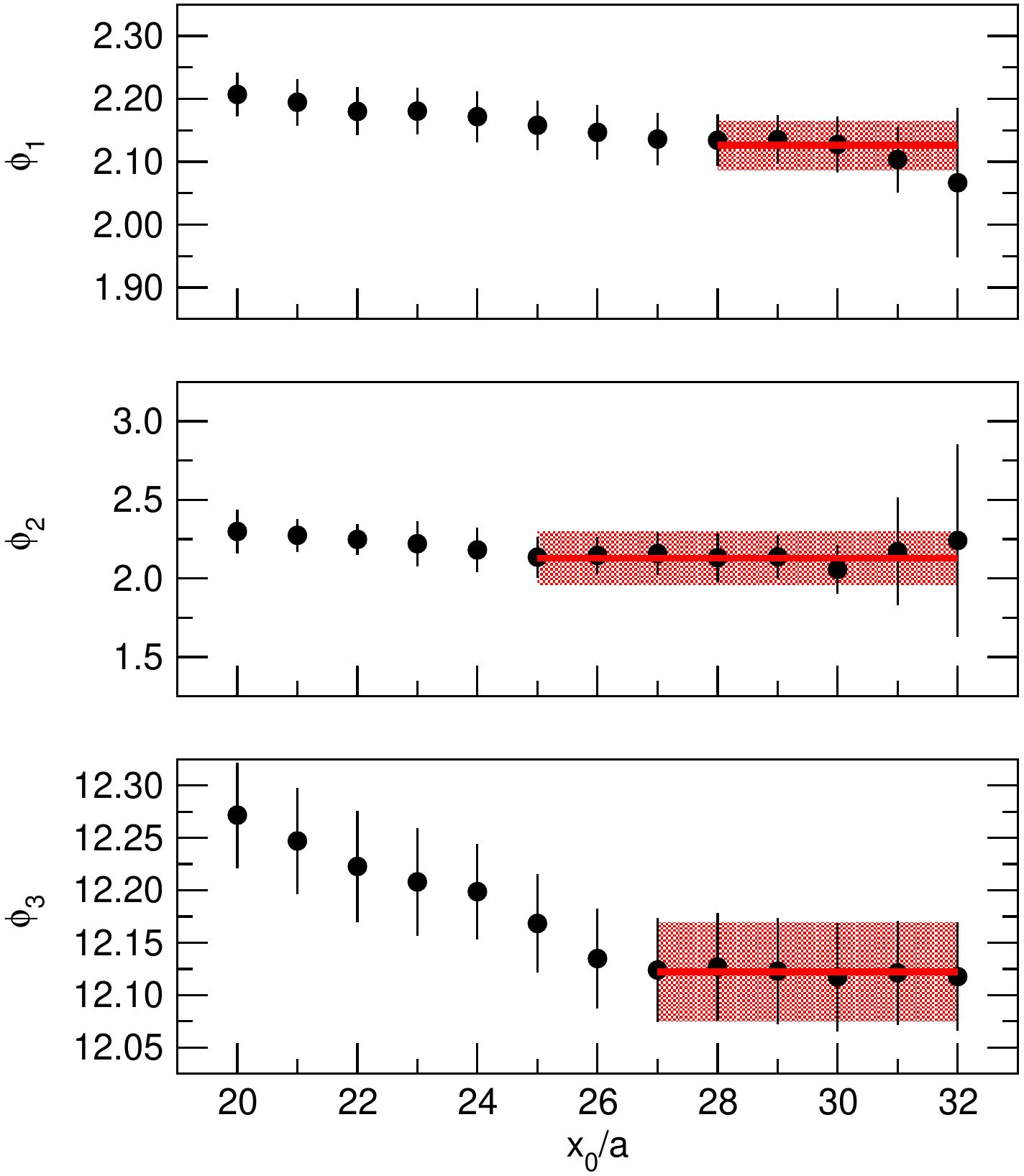}
      \end{minipage}
      \caption{%
      Meson effective masses and effective $\phi$ observables for the ensemble
      \texttt{A380a07b324+RW1}, together with the selected plateaux and the fits
      to a constant. Values in MeV are obtained by using the reference value $(8
      t_0)^{1/2} = 0.415 \text{ fm}$~\cite{Bruno:2016plf}.
      }
      \label{fig:effmasses-1}
   \end{center}
\end{figure}

\begin{figure}
   \begin{center}
      \begin{minipage}[c]{.625\textwidth}
         \includegraphics[width=\textwidth]{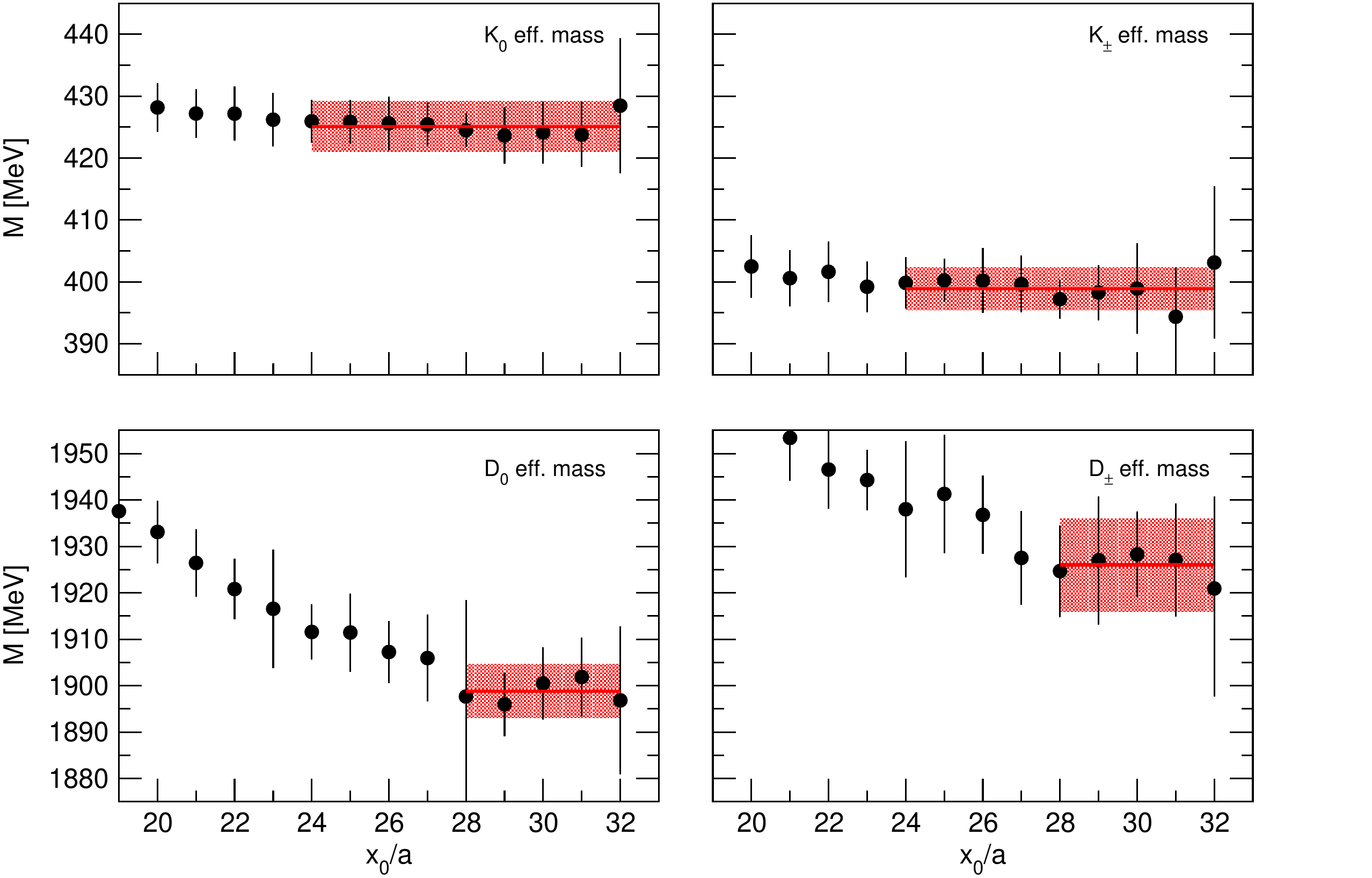}
      \end{minipage}\begin{minipage}[c][6cm][c]{.35\textwidth}
         \includegraphics[width=\textwidth]{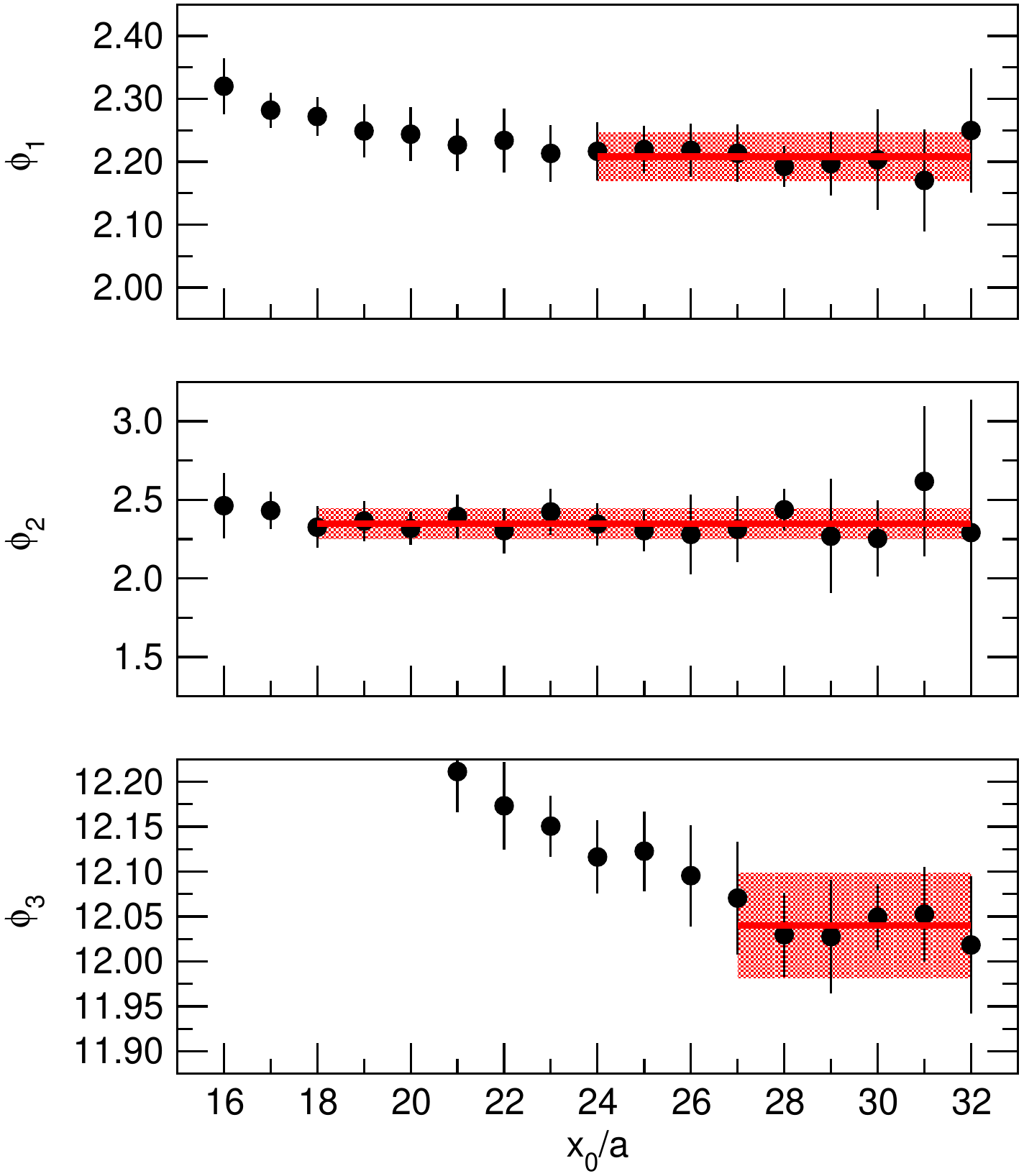}
      \end{minipage}
      \caption{%
      Meson effective masses and effective $\phi$ observables for the ensemble
      \texttt{A360a50b324+RW2}, together with the selected plateaux and the fits
      to a constant. Values in MeV are obtained by using the reference value $(8
      t_0)^{1/2} = 0.415 \text{ fm}$~\cite{Bruno:2016plf}.
      }
      \label{fig:effmasses-2}
   \end{center}
\end{figure}

\subsection{Baryon masses}\label{sec:details:baryon}
% !TEX root = paper.tex

Baryon interpolating operators are written in terms of Gaussian-smeared
fermion fields defined by
\begin{gather}
   \Psi_{\text{(s)}} = ( 1 + \omega H)^n \Psi \ ,
   \qquad
   \bar{\Psi}_{\text{(s)}} = \bar{\Psi} ( 1 + \omega H)^n \ ,
\end{gather}
where $\Psi$ and $\bar{\Psi}$ are the dressed fermion fields defined in
eq.~\eqref{eq:dressed-psi}, $\omega$ and $n$ are adjustable parameters, and $H$
is the spatial hopping operator given by
\begin{gather}
   H \Psi(x) = \sum^{3}_{k=1}
   \left\{ V(x,k) \Psi(x+a\hat{e}_k)+ V(x-a\hat{e}_k,k)^\dag \Psi(x-a\hat{e}_k) \right\}
   \ .
\end{gather}
In this formula $V$ is an SU(3) smeared link variable. In practice, we construct
$V$ by means of a generalization of the gradient flow restricted to a single
time slice, i.e. we solve a dicretized version of the following differential
equation
\begin{gather}
   \partial_s U_s(x,k) = \partial_{x,k} \sum_{i \neq j} \tr P_{s,ij}(x)
   \ ,
\end{gather}
where $P_{s,ij}(x)$ is the plaquette in $x$ on the plane identified by the
indices $(i,j)$, constructed with the field $U_s$, and the initial condition
$U_0 = U$ is used. The smeared field $V$ is identified with $U_s$ at the chosen
maximum value of the auxiliary flowtime $s$. We notice that the smeared fields
are local in time and invariant under U(1) gauge transformations.

For definiteness we consider the chiral representation of the gamma matrices and
$C=i \gamma_0 \gamma_2$. The C-even zero-momentum interpolating operators for
spin-1/2 baryons considered in this work can be all written in the form
\begin{gather}
   \label{eq:baryon:intop:B}
   B(x_0) = \sum_{\vec{x}} \sum_{\substack{abc\\fgh}}
   \epsilon_{abc} F_{fgh} 
   \\ \hspace{18mm} \nonumber 
   \times \left\{ \Psi_{\text{(s)}fa} \Psi_{\text{(s)}gb}^t C \gamma_5 \Psi_{\text{(s)}hc}(x_0,\vec{x})
   - C \bar{\Psi}^t_{\text{(s)}fa} \bar{\Psi}_{\text{(s)}gb} C \gamma_5 \bar{\Psi}^t_{\text{(s)}hc}(x_0,\vec{x}) \right\}
   \ ,
\end{gather}
where $a,b,c$ are colour indices and $f,g,h$ are flavour indices (spin indices
are implicit or contracted).
% %
% \begin{gather}
%    B(x_0) = \sum_{\vec{x}} %\sum_{\substack{abc\\fgh}}
%    \epsilon_{abc} F_{fgh} \frac{1+\gamma_0}{2} \left\{ \Psi_{\text{(s)}fa} \Psi_{\text{(s)}gb}^t C \gamma_5 \Psi_{\text{(s)}hc}
%    - C \bar{\Psi}^t_{\text{(s)}fa} \bar{\Psi}_{\text{(s)}gb} C \gamma_5 \bar{\Psi}^t_{\text{(s)}hc} \right\}
%    \ ,
% \end{gather}
% %
% where the coulor indices $a,b,c$ and the flavour indices and $f,g,h$ are summed
% over, and the spin indices are implicit or contracted. Different baryons are
Different baryons are obtained by choosing particular tensors $F_{fgh}$,
according to the table~\ref{tab:baryons1/2}. In this case the zero-momentum
two-point function is defined as
\begin{gather}
   C(x_0) = \langle B^t(0) C \frac{1 + \gamma_0}{2} B(x_0) \rangle
   \ .
\end{gather}
%
% where the second equality follows from invariance under Euclidean time
% reversal.\footnote{%
% %
% Euclidean time reversal is a symmetry with C$^\star$ boundary conditions if
% defined with an unconventional phase for the transformation of the fermion
% fields $\psi(x_0,\vec{x}) \to i \gamma_0 \gamma_5 \psi(-x_0,\vec{x})$ and
% $\bar{\psi}(x_0,\vec{x}) \to -i \bar{\psi}(-x_0,\vec{x}) \gamma_5 \gamma_0$,
% while the gauge fields transform as usual.
% %
% } In practice...
%
For the $\Omega^-$ baryon we use the following C-even zero-momentum
interpolating operator
\begin{gather}
   \Omega_{j}(x_0) =  \sum_{\vec{x}} \sum_{abc} \sum_k
   \epsilon_{abc} \left( \delta_{jk} - \tfrac{1}{3} \gamma_j \gamma_k \right)
   \\ \hspace{18mm} \nonumber 
   \times \left\{ S_{\text{(s)}a} S_{\text{(s)}b}^t C \gamma_k S_{\text{(s)}c}(x_0,\vec{x})
   - C \bar{S}^t_{\text{(s)}a} \bar{S}_{\text{(s)}b} \gamma_k C \bar{S}^t_{\text{(s)}c}(x_0,\vec{x}) \right\}
   \ ,
\end{gather}
where $S_{\text{(s)}}= \Psi_{\text{(s)}s}$ is the smeared dressed field of the
strange quark. In this case the zero-momentum two-point function is defined as
\begin{gather}
   C(x_0) = \sum_j \langle \Omega^t_j(0) C \frac{1 + \gamma_0}{2} \Omega_j(x_0) \rangle
   \ .
\end{gather} 

In practice we use point sources to calculate the needed two-point functions.
For each configuration we construct $12h$ point sources (running over each color
and spin index, and located at $h$ random timeslices), where $h=8$ for all
ensembles except \texttt{A500a50b324} and \texttt{A360a50b324+RW2} for which we
have used $h=4$. We used the Generalized Eigenvalue
Problem~\cite{Guagnelli:1998ud,Luscher:1990ck} to optimize the smearing
parameters. The results given in this paper use a smearing $s = 3.6$, $n=400$
and $\omega=0.5$ for the source and no smearing for the sink. Because of the
boundary conditions quark-quark and antiquark-antiquark Wick contractions do not
vanish (but are exponentially suppressed with the volume). In this work we have
simply neglected these Wick contractions which effectively means that we are
calculating the two-point function of some partially-quenched baryons as
discussed in section~\ref{sec:overview:masses}. We plan to quantify the
contribution of the extra Wick contractions in future work.

\begin{table}
   \begin{center}
      \begin{tabular}{cc}
         $(1/2)^+$ Baryon & Non-zero components of $F_{fgh}$ \\
         \hline
         p & $F_{uud} = 1$ \\
         n & $F_{ddu} = 1$ \\
         $\Lambda_0$ & $F_{sud} = 2$, $F_{dus} = 1$, $F_{uds} = -1$ \\
         $\Sigma^+$ & $F_{uus} = 1$ \\
         $\Sigma^-$ & $F_{dds} = 1$ \\
         $\Xi_0$ & $F_{ssu} = 1$ \\
         $\Xi^-$ & $F_{ssd} = 1$ \\
         \hline
      \end{tabular}
   \end{center}
   \caption{%
   Flavour tensor $F_{fgh}$ defining the interpolating operators for spin-$1/2$
   baryons via eq.~\eqref{eq:baryon:intop:B}. The flavour indices can take
   values $u$, $d$, $s$, $c$.
   }
   \label{tab:baryons1/2}
\end{table}

Given the zero-momentum two-point function $C(x_0)$, we define the effective
mass $M(x_0)$ simply as:
\begin{gather}
   M(x_0)= 
   \frac{1}{a} \log \frac{C(x_0)}{C(x_0+a)}
   \ .
\end{gather}
The effective mass $M_c(x_0)$ corrected for structure-independent finite-volume
effects is defined as for the mesons using eq.~\eqref{eq:corrected-effmass}. A
selection of effective masses and mass differences with the corresponding
plateau fits are shown in figures~\ref{fig:baryons:effmasses-1} and
\ref{fig:baryons:effmasses-2}.

\begin{figure}
   \begin{center}
      \includegraphics[width=\textwidth]{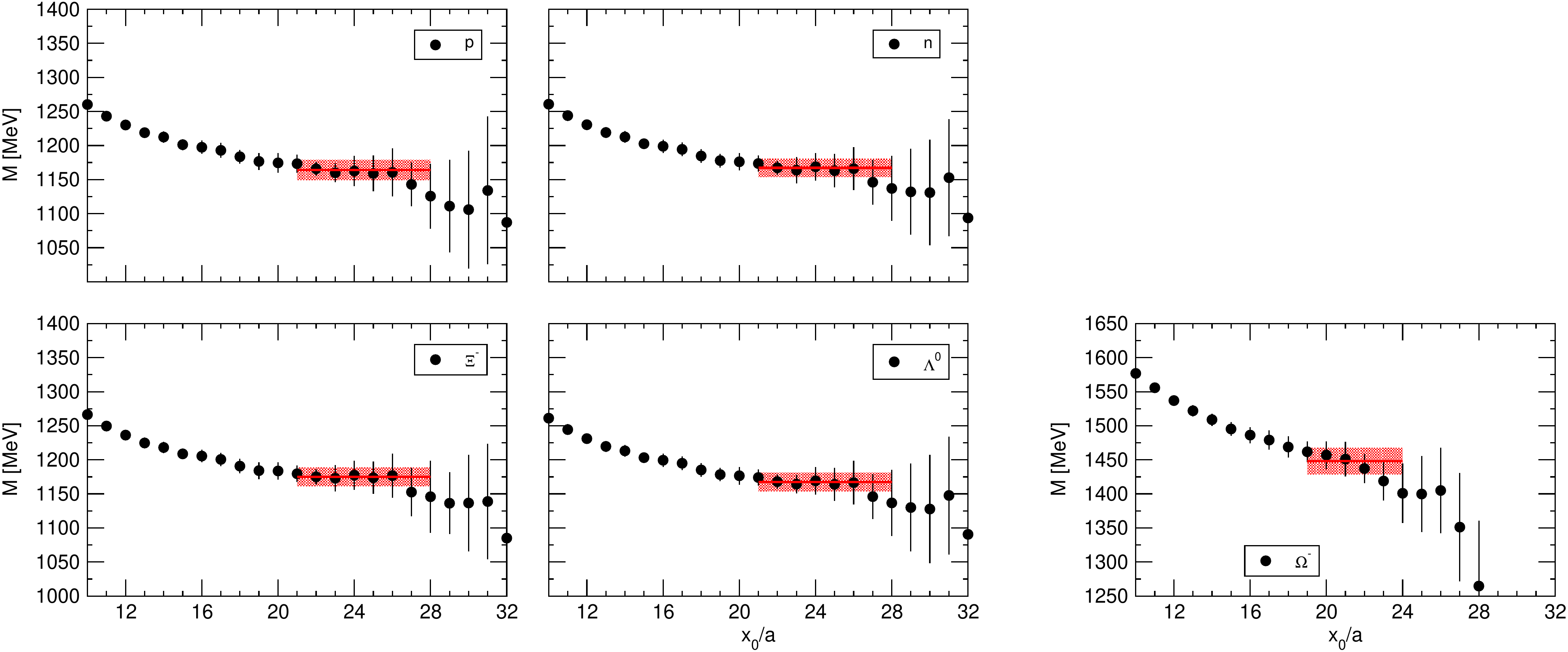}
      
      \vspace{3mm}
      
      \includegraphics[width=\textwidth]{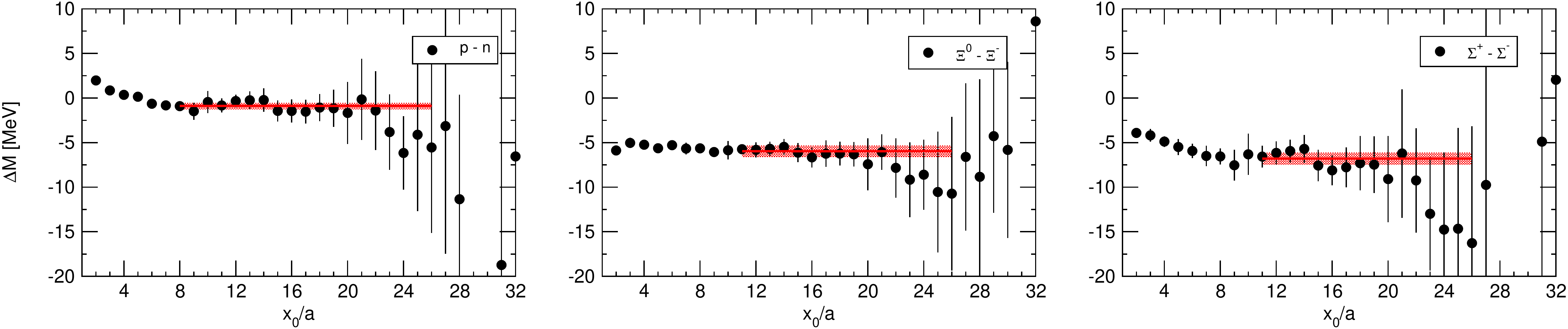}      
      \caption{%
      Baryon effective masses for the ensemble \texttt{A380a07b324+RW1}, with
      the selected plateaux and the fits to a constant. Values in MeV are
      obtained by using the reference value $(8 t_0)^{1/2} = 0.415 \text{
      fm}$~\cite{Bruno:2016plf}.
      }
      \label{fig:baryons:effmasses-1}
   \end{center}
\end{figure}

\begin{figure}
   \begin{center}
      \includegraphics[width=\textwidth]{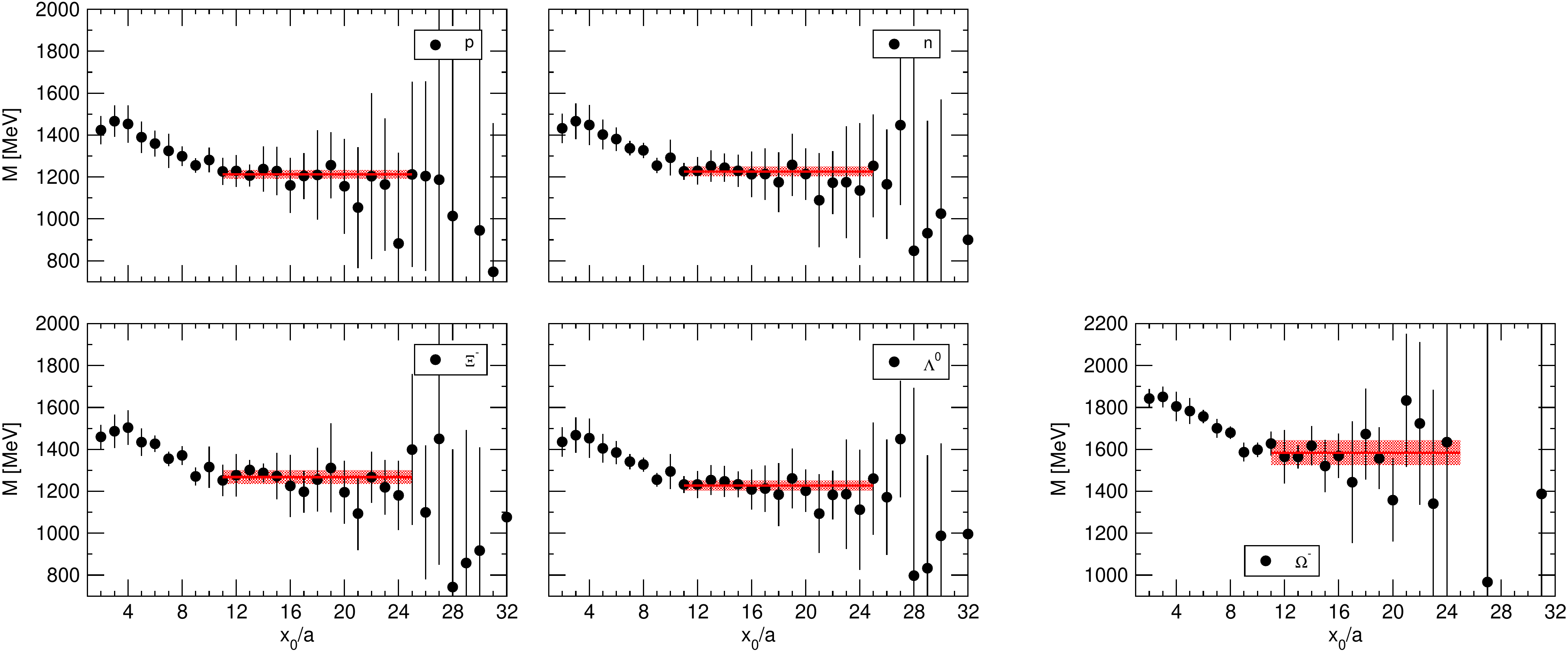}
      
      \vspace{3mm}
      
      \includegraphics[width=\textwidth]{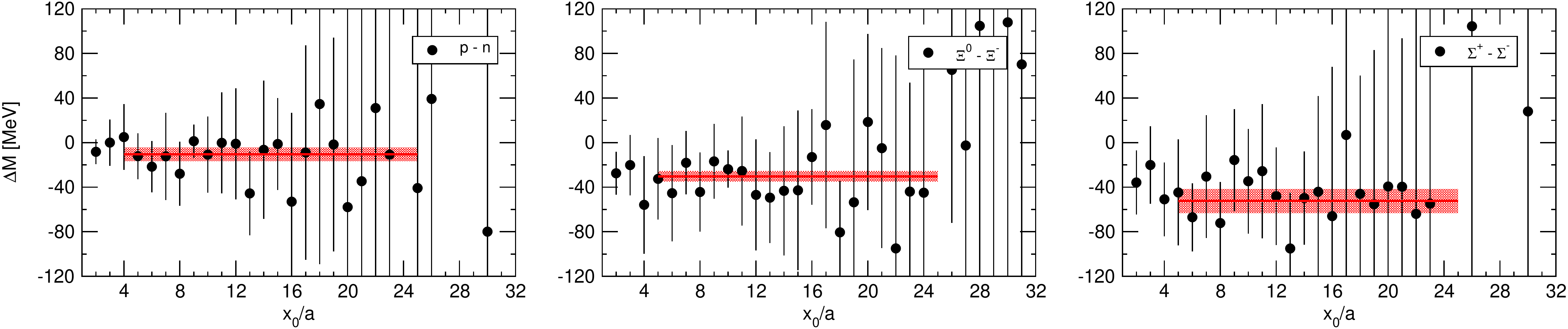}      
      \caption{%
      Baryon effective masses and effective $\phi$ observables for the ensemble
      \texttt{A360a50b324+RW2}, with the selected plateaux and the fits to a
      constant. Values in MeV are obtained by using the reference value $(8
      t_0)^{1/2} = 0.415 \text{ fm}$.
      }
      \label{fig:baryons:effmasses-2}
   \end{center}
\end{figure}

\subsection{Tuning strategy}\label{sec:details:tuning}
% !TEX root = paper.tex

For fixed values of $\alpha$ and $\beta$, the bare quark masses need to be tuned
to obtain the desired values of the $\phi$ variables given by
eq.~\eqref{eq:Uline}. Since we have chosen to work with $m_d = m_s \equiv
m_{ds}$, this is a three-parameter tuning problem. In practice, we have followed
the following steps:
\begin{enumerate}
   \item Generate some ensembles with smaller statistics ($\sim$200 thermalized
   configurations), and get a rough estimate $\hat{m}^{(0)} = (m_u^{(0)},
   m_{ds}^{(0)}, m_c^{(0)})$ for the quark masses.
   \item Generate an ensemble with full statistics (at least 1000 thermalized
   configurations) with quark masses equal to $\hat{m}^{(0)}$. Calculate the
   values of the $\phi^{(0)} = (\phi^{(0)}_1 , \phi^{(0)}_2, \phi^{(0)}_3)$
   observables on these configurations.
   \item Choose three new sets of quark masses $\hat{m}^{(i)}$ with $i=1,2,3$
   fairly close to $\hat{m}^{(0)}$, and calculate the values of the $\phi^{(i)}$
   observables corresponding to these quark masses, by means of mass
   reweighting. Find the tuned values of the quark masses $\hat{m}^{(t)}$ by
   linear interpolation, i.e. by assuming that the $\phi$ observables depend on
   the masses $\hat{m}$ as in $\phi = A \hat{m} + b$, where $A$ is a $3 \times
   3$ matrix and $b$ is a $3$-vector. A few attempts may be necessary in order
   to find values for $\hat{m}^{(i)}$ for which the reweighting does not have an
   overlap problem, and for which the tuned value is found either by
   interpolation or by a mild extrapolation.
   \item Generate an ensemble with full statistics (2000 thermalized
   configurations) with quark masses equal to $\hat{m}^{(t)}$. Calculate the
   values of the $\phi^{(t)}$ observables on these configurations.
   \item If the extrapolation in point 2 is too long, then one does not get the
   target value for the $\phi$ observables, and one needs to repeat everything
   from step 2 with $\hat{m}^{(0)} \leftarrow \hat{m}^{(t)}$. On the other hand,
   some residual small mistuning due to the linear approximation is corrected by
   repeating step 2 with $\hat{m}^{(0)} \leftarrow \hat{m}^{(t)}$. In this case
   a corrected tuned value $\hat{m}^{(t)}$ is found, and the observables are
   calculated by mass reweighting to the corrected tuned value, without
   generating new configurations.
\end{enumerate}
In this paper we describe all ensembles with full statistics
that we have generated in the tuning procedure, but not all intermediate mass
reweighting factors.

It is worth noticing that we have tried carrying out step 2 by changing only the
valence quark masses, i.e. without including a mass reweighting factor. However,
we usually incurred into a problem of overshooting which rendered this strategy
unusable. Nevertheless, if one starts from a value of $\hat{m}^{(0)}$ which is
far away from the target value, a first tuning iteration performed by changing
only the valence quark masses can be a relatively inexpensive way to move
towards the correct region of parameter space.

\subsection{Statistical analysis}\label{sec:details:analysis}
% !TEX root = paper.tex

Errors are calculated with the gamma method in the particular incarnation
of~\cite{Schaefer:2010hu}. The integrated autocorrelation time is calculated
first with the Wolff's automatic windowing procedure~\cite{Wolff:2003sm} with
parameter $S=1$. Among all considered observables $t_0/a^2$ has the largest
integrated autocorrelation time and we use this as an estimate of the
exponential autocorrelation time $\tau_{\text{exp}}$, which is a property of the
particular Markov chain rather than of the observable. For each observable the
autocorrelation function $\Gamma(t)$ is calculated from data for
$t=0,1,\dots,\bar{t}$, where $\bar{t}$ is chosen in such a way that the central
value of $\Gamma(t)$ is positive for any $t \le \bar{t}$ and negative for
$t=\bar{t}+1$. Then the autocorrelation function is extended for $t > \bar{t}$
with a single exponential $\Gamma(t) = \Gamma(\bar{t}) \exp
[-(t-\bar{t})/\tau_{\text{exp}} ]$. The extended autocorrelation function is
used in the gamma method to calculate errors and integrated autocorrelation
times. The analysis has been carried out with our own code, which implements
the ideas of~\cite{Ramos:2018vgu}.

\subsection{Sign of the pfaffian}\label{sec:details:sign}
% !TEX root = paper.tex

Given a quark field $\psi$, we introduce the corresponding antiquark field
$\psi^{\mathcal{C}}=C^{-1}\bar{\psi}^t$, where the charge-conjugation matrix $C$
can be chosen to be $i\gamma_0\gamma_2$ in the chiral basis. C$^\star$ boundary
conditions for the fermion fields can be written as
\begin{gather}
    \begin{pmatrix} \psi(x+L\hat{k})\\
    \psi^\mathcal{C}(x+L\hat{k})
    \end{pmatrix}
    =
    \begin{pmatrix} \psi^\mathcal{C}(x)\\
    \psi(x)
    \end{pmatrix}
    \equiv
    K\begin{pmatrix} \psi(x)\\
    \psi^\mathcal{C}(x)
    \end{pmatrix}
    \ .
\end{gather}
With C$^\star$ boundary conditions the Dirac operator $D$ acts on the
quark-antiquark doublet in a non-diagonal way, and it is therefore a $24 V
\times 24 V$ matrix, where $V = TL^3/a^4$. An explicit representation for the
Dirac operator is given in appendix~\ref{app:pfaffian}. The integration of a
quark field in the path integral yields the pfaffian $\pf (CKD)$ in place of the
standard fermionic determinant, where $CKD$ is an antisymmetric matrix (see
proposition~\ref{prop:antisymmetry} in appendix~\ref{app:pfaffian}). The
absolute value of the pfaffian is given by
\begin{gather}
    \left| \pf (CKD) \right|
    =
    \left| \det D \right|^{1/2}
    \ ,
    \label{eq:power1/4}
\end{gather}
and can be simulated by means of a standard RHMC algorithm, while the sign of
the pfaffian can be incorporated as a reweighing factor.

In order to calculate this sign, it is convenient to relate the pfaffian to the
spectrum of the hermitian Dirac operator $Q=\gamma_5 D$. It turns out that the
spectrum of $Q$ is doubly degenerate (see proposition~\ref{prop:degeneracy} in
appendix~\ref{app:pfaffian}). Let $\lambda_{n=1,\dots,12V} \in \mathbb{R}$ be
the list of eigenvalues of $Q$, each of them appearing a number of times equal
to half their degeneracy. The following simple formula holds (see
proposition~\ref{prop:pfaffian} in appendix~\ref{app:pfaffian}):
\begin{gather}
    \pf (CKD)
    =
    \prod_{n=1}^{12V} \lambda_n
    \ . \label{eq:pf-lambda}
\end{gather}
It follows that the pfaffian is positive (resp. negative) if the number of
negative eigenvalues $\lambda_n$ is even (resp. odd). In practice, we calculate
the sign by following the eigenvalue flow as a function of the quark mass $m$.
We use the crucial fact that the eigenvalues of $Q(m)$ can be labeled in such a
way that they are continuous functions of $m$. As $m$ is continuously varied,
the pfaffian flips sign every time a degenerate pair of eigenvalues of $Q(m)$
crosses zero. It follows that
\begin{gather}
   \sgn \pf [CKD(m)] = (-1)^{c(m,M)} \sgn \pf [CKD(M)]
   \label{eq:pf-relative-sign}
   \ ,
\end{gather}
where we have highlighted the mass dependence of the Dirac operator, and
$c(m,M)$ is the number of degenerate pairs of eigenvalues of $Q(m)$ crossing zero
as the mass is continuously varied from $m$ to $M$.

If $M$ is chosen very large, then $Q(M)$ is approximately equal to $M \gamma_5$
and the number of negative $\lambda_n$'s is even and equal to $6V$. Hence, the
pfaffian of $CKD(M)$ is positive, and eq.~\eqref{eq:pf-relative-sign} implies
that the sign of the pfaffian of $CKD(m)$ can be calculated by counting the
number of eigenvalue pairs crossing zero between the target mass $m$ and some
large mass $M$: the pfaffian is positive if this number is even and negative
otherwise.

In the particular case of QCD+QED simulations with 4 flavours we use the
following observation to eliminate the arbitrariness associated to the choice of
$M$. If $D_{+2/3}(m)$ is the Dirac operator for a quark with electric charge
$+2/3$, then the Dirac operators for the up and charm quarks are simply $D_u =
D_{+2/3}(m_u)$ and $D_c = D_{+2/3}(m_c)$, since the two quarks differ only for
their mass. Using eq.~\eqref{eq:pf-relative-sign}, one gets the contribution to
the sign of up and charm quarks as:
\begin{gather}
   \sgn \pf (CKD_u) \, \sgn \pf (CKD_c) = (-1)^{c_{+2/3}(m_u,m_c)}
   \label{eq:pf-sign-uc}
   \ ,
\end{gather}
where the subscript of $c_{+2/3}$ stresses the fact that we need to count the
eigenvalue crossings of the hermitian Dirac operator for a quark with electric
charge $+2/3$. Analogously, the contribution to the sign of down and
strange quarks is:
\begin{gather}
   \sgn \pf (CKD_d) \, \sgn \pf (CKD_s) = (-1)^{c_{-1/3}(m_d,m_s)}
   \label{eq:pf-sign-ds}
   \ .
\end{gather}
The reweighting factor needed to account for the sign of the fermionic pfaffian
is then given by:
\begin{gather}
   W_{\text{sgn}} = \prod_{f=u,d,s,c} \sgn \pf (CKD_f) = (-1)^{c_{+2/3}(m_u,m_c)} (-1)^{c_{-1/3}(m_d,m_s)}
   \label{eq:pf-sign}
   \ .
\end{gather}
At the U-symmetric point $m_d=m_s$ one has trivially $c_{-1/3}(m_d,m_s) = 0$,
and one only needs to count the number of eigenvalue crossings for up-type
quarks as the mass is varied from the up-quark mass to the charm-quark mass.

We give a brief account of the techniques used in this work to calculate
$c(m,M)$, i.e. to count the number of eigenvalue crossings of $Q$ as the mass is
varied continuously form $m$ to $M$. Our method is based on two stages:
\begin{enumerate}
   
   \item A first fast algorithm finds an interval $I \subset [m,M]$ in the
   considered range of bare masses, with the property that no sign flip occurs
   in the complement of $I$. Typically $I$ is much smaller in length than the
   original $[m,M]$, and for most configurations $I$ turns out to be empty.
   
   \item If $I$ is not empty, we follow the flow of a certain number of
   eigenvalues which are the closest to zero, only inside the typically small
   interval $I$, using the methods described
   in~\cite{Campos:1999du,Mohler:2020txx}. By tracking the eigenvalues as
   functions of the mass, we can determine whether a sign flip occurs.

\end{enumerate}
To the best of our knowledge, the first step of our method has never been used
in similar calculations. Since it speeds up significantly the calculation of the
sign of fermionic pfaffian, and it can be used also in different contexts (e.g.
for the calculation of the sign of the fermionic determinant in QCD, or the sign
of the fermionic pfaffian in gauge theories with adjoint fermions), we discuss
it here in some detail.

The starting point is the following simple observation. Let $\mu(m)$ be the
smallest eigenvalue of $|Q(m)|$, i.e.
\begin{gather}
   \mu(m) = \min_n |\lambda_n(m)| \ .
\end{gather}
If $\bar{\mu} \equiv \mu(\bar{m}) > 0$, then no eigenvalue $\lambda_n(m)$ flips
sign as long as $m$ is in the interval defined by $|m-\bar{m}|<\bar{\mu}$.

\begin{proof}
   
   We use the notation $m = \bar{m} + \delta m$. Let $v$ be a normalized vector.
   Using the identity $Q(m) = Q(\bar{m}) + \delta m \, \gamma_5$, the triangular
   inequality and the unitarity of $\gamma_5$, one readily derives
   \begin{gather}
      \left\| Q(m) v \right\|
      \ge
      \left\| Q(\bar{m}) v \right\| - |\delta m| \, \left\| \gamma_5 v \right\|
      =
      \left\| Q(\bar{m}) v \right\| - |\delta m|
      \ge
      \bar{\mu} - |\delta m|
      \ .
      \label{eq:pf:bound}
   \end{gather}
   In the last step we have also used the fact that $\bar{\mu}$ is the smallest
   eigenvalue of $|Q(\bar{m})|$. If we choose for $v$ any of the eigenvectors of
   $Q(m)$, the above inequality specializes to
   \begin{gather}
      | \lambda_n(m) | \ge \bar{\mu} - |\delta m|
      \ .
   \end{gather}
   If $|m-\bar{m}| = |\delta m| < \bar{\mu}$ then $| \lambda_n(m) | > 0$. Since
   $\lambda_n(m)$ is continuous in $m$, it is either always positive or always
   negative in the interval defined by $|m-\bar{m}| < \bar{\mu}$.
   
\end{proof}

This result allows us to design the following iterative algorithm which
restricts the original interval $I = [m, M]$, in which we search for sign flips,
to the interval $I' = [m',M]$ constructed in the following way:
\begin{enumerate}
   
   \item Set $m_0 = m$ and $n=0$.
   
   \item Define $\mu_n$ to be a safe non-negative lower bound for $\mu(m_n)$.
   The lowest eigenvalue of $|Q(m_n)|$ can be calculated with a standard power
   method applied to $Q(m_n)^{-2}$.
   
   \item Since no sign flip occurs in $[m_n,m_n+\bar{\mu}_n)$, we set $m_{n+1} =
   m_n + (1-\epsilon) \bar{\mu}_n$. $\epsilon$ is a tunable small positive
   parameter which we choose to be 0.1.
   
   \item If $m_{n+1}-m_n < \eta m_n$ or $m_{n+1} > M$, we set $m' = m_{n+1}$ and
   we stop the algorithm, otherwise we repeat from 2 with $n \leftarrow n+1$.
   $\eta$ is a tunable parameter satisfying $0 < \eta < 1$, which we choose to
   be 1/4.
   
\end{enumerate}
For most configurations, $\mu_n$ is an increasing sequence. In this situation
this part of the algorithm covers the whole interval in a handful of steps, as
illustrated in figure~\ref{fig:pf:scan}. If the interval $I' = [m',M]$ is not
empty, we restrict it to a new interval $I'' = [m',M']$ constructed in the
following way:
\begin{enumerate}
   
   \item Set $M_0 = M$ and $n=0$.
   
   \item Define $\mu_n$ to be a safe non-negative lower bound for $\mu(M_n)$.
   
   \item Since no sign flip occurs in $(M_n-\bar{\mu}_n,M_n]$, we set $M_{n+1} =
   M_n - (1-\epsilon) \bar{\mu}_n$.
   
   \item If $M_n-M_{m+1} < \eta M_n$ or $M_{n+1} < m'$, we set $M' = M_{n+1}$
   and we stop the algorithm, otherwise we repeat from 2 with $n \leftarrow
   n+1$.
   
\end{enumerate}
If the interval $I'' = [m',M']$ is not empty, then one resorts to the methods
described in~\cite{Campos:1999du,Mohler:2020txx}. In this case, one tracks the
mass-dependence of eigenvalues and eigenvectors in the interval $I''$. Since a
fairly accurate determination of eigenvectors is required for a fairly fine scan
of the mass, this stage of the algorithm is intrinsically more expensive than
the previous one. However, in practice, it needs to be applied only to a
relatively small number of configurations and to small mass intervals.

\begin{figure}
   \begin{center}
      \includegraphics[width=.8\textwidth]{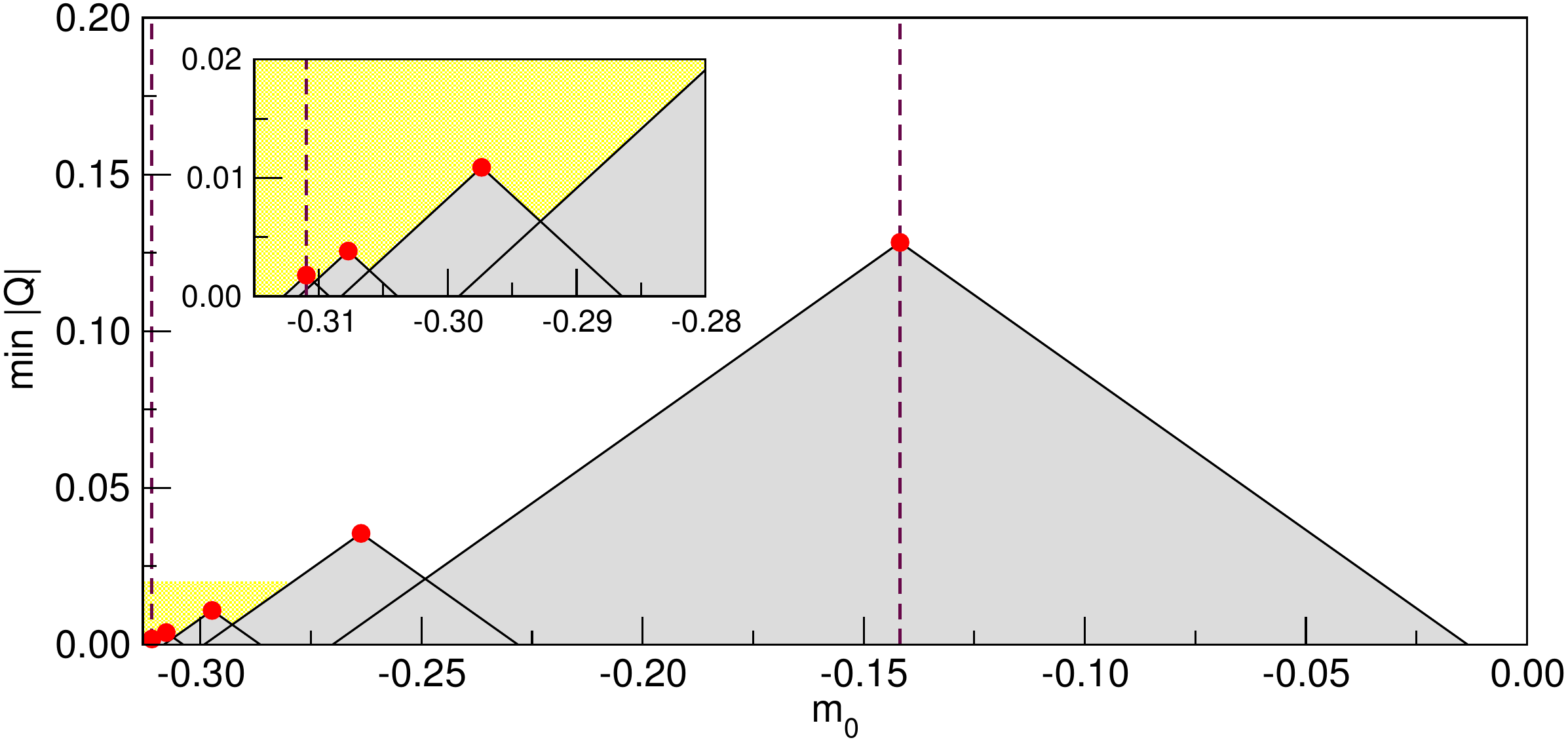}
      \caption{%
      Smallest eigenvalue of $|Q|$ for up-type quarks as a function of the
      valence mass $m_0$ (red points), calculated on a representative
      configuration (\texttt{C380a50b324} ensemble). Using the
      bound~\eqref{eq:pf:bound}, one proves that $|Q|$ has no eigenvalues in the
      grey areas. This implies that no eigenvalue of $Q$ crosses zero between
      the up-quark mass (left dashed vertical line) and the charm-quark mass
      (right dashed vertical line). One can see that we are able to flow the
      lowest eigenvalue across two orders of magnitude in only five steps. The
      inset is a zoom-in of the yellow area.
      }
      \label{fig:pf:scan}
   \end{center}
\end{figure}

\subsection{Algorithmic parameters}\label{sec:details:algo}
% !TEX root = paper.tex

In this section we report on the algorithmic parameters chosen to generate the
presented gauge configurations, with a particular emphasis on integration
scheme, solvers and rational approximations.

\subsubsection*{Rational approximations}

Because of C$^\star$ boundary conditions, we need to simulate the following
fermionic determinant
\begin{gather}
   \det (\hat{D}_u^\dag \hat{D}_u)^{1/4} \det (\hat{D}_{ds}^\dag \hat{D}_{ds})^{1/2} \det (\hat{D}_c^\dag \hat{D}_c)^{1/4}
\end{gather}
for the QCD+QED ensembles ($\alpha \neq 0$, $m_d=m_s$) and the following
fermionic determinant
\begin{gather}
    \det (\hat{D}_{uds}^\dag \hat{D}_{uds})^{3/4} \det (\hat{D}_c^\dag \hat{D}_c)^{1/4}
\end{gather}
for the QCD ensembles ($\alpha = 0$, $m_u=m_d=m_s$). Here $\hat{D}$ denotes
the even-odd preconditioned Dirac operator. The inverse operators
$(\hat{D}^\dag \hat{D})^{-\gamma}$ with $\gamma=1/4,1/2,3/4$ are approximated
with rational functions of $\hat{D}^\dag \hat{D}$.

In practice, we construct the rational function $R(x)$ of order $(N,N)$ that
minimizes the relative precision, i.e. the maximum of $| 1 - x^{\gamma} R(x) |$
over some interval $x \in [r_a^2,r_b^2]$. The approximation range is chosen in
such a way that the eigenvalues of $\hat{D}^\dag \hat{D}$ are included in
$[r_a^2,r_b^2]$ most of the time. In table~\ref{tab:algo-param} we report the
parameters defining the rational approximations used for this work.

The error introduced by the rational approximation is corrected by means of a
reweighting factor. This strategy is identical to the one adopted in the
\texttt{openQCD} code~\cite{openQCD}. The generalization to any value of
$\gamma$ of the reweighting factor can be found in~\cite{Campos:2019kgw}. The
parameters of the rational functions have been chosen in such a way that the
reweighting factor does not introduce any detectable increase in the errors of
the considered observables.

When a rational approximation is used, the Dirac operator always appears in the
pseudofermion actions in the combination $\hat{D}^\dag \hat{D} +  \mu^2$ for
strictly positive values of $\mu$. The rational approximation has the effect of
removing the infinite potential barrier encountered when the fermionic
determinant or pfaffian tries to change sign~\cite{Mohler:2020txx}, pretty much
in the same way as the twisted-mass reweighting procedure proposed
in~\cite{Luscher:2012av}. However, if the ratio $r_b/r_a$ becomes too large or
if the precision of the rational approximation is chosen to be too high, then
the smallest twisted mass $\mu$ becomes too small, and the potential barrier may
become large enough to jeopardize ergodicity and stability of the algorithm. For
this reason we have progressively reduced the precision of the rational
approximations used in our runs, settling for a precision of a few units of
$10^{-6}$ in our latest runs.

\subsubsection*{Pseudofermion actions and solvers}

For a given flavour, the approximated determinant $\det R^{-1}$ is represented
by means of a sum of pseudofermion actions in the following way. We start from
the rational function
\begin{gather}
   R = A \prod_{k=1}^N \frac{\hat{D}^\dag \hat{D} + \nu_k^2}{\hat{D}^\dag \hat{D} + \mu_k^2} \ .
\end{gather}
In the following we assume the ordering $\nu_1 < \nu_2 < \dots < \nu_N$ and
$\mu_1 < \mu_2 < \dots < \mu_N$. We introduce a pseudofermion action for each
factor in the above equation with $k=1,\dots,N_2$, and a single pseudofermion
action for the remaining factors, i.e.
\begin{gather}
   S_k = \phi_k^\dag \frac{\hat{D}^\dag \hat{D} + \nu_k^2}{\hat{D}^\dag \hat{D} + \mu_k^2} \phi_k
   \ , \qquad \text{for } k=1,\dots,N_2
   \ , \\
   S_{N_2+1} = \phi_{N_2+1}^\dag \left( \prod_{k=N_2+1}^N \frac{\hat{D}^\dag \hat{D} + \nu_k^2}{\hat{D}^\dag \hat{D} + \mu_k^2} \right) \phi_{N_2+1}
   \ .
\end{gather}
In practice, the pseudofermion actions are represented by means of partial
fraction decompositions as explained in~\cite{Campos:2017fly}. The chosen values
for $N_2$ are reported in table~\ref{tab:algo-param}.

One needs to invert the operator $(\hat{D} + i \mu_k)$ in the calculation of the
pseudofermion actions, and the operator $(\hat{D} + i \nu_k)$ in the generation
of the pseudofermion fields. The multishift conjugate gradient is used for the
pseudofermion action $S_{N_2+1}$, while a deflated generalized conjugate
residual method~\cite{Luscher:2007se}, preconditioned with the Schwarz
alternating procedure~\cite{Luscher:2003vf}, is used for all other pseudofermion
actions. We use a $10^{-8}$ residue for all solvers used in the calculation of
the force, and a $10^{-10}$ residue for all solvers used in the calculation of
the action and in the generation of the pseudofermion fields.

A distinctive feature of the QCD+QED simulations is the need for different
deflation subspaces for different values of the electric charge ($\hat{q}=2,4$
in our simulations). Using the notation of the \texttt{openQ*D} input
files~\cite{Campos:2017fly}, we used $\texttt{mu} = 0.001$ for the $64 \times
32^3$ lattices and $\texttt{mu} = 0.005$ for the other ones, $\texttt{nkr} \ge
24$, and a total of $\texttt{Ns} = 20$ deflation vectors. The size of the
deflation blocks have been chosen to be as large as possibile, compatibly with
the size of the local lattice and the constraints of the simulation code.

\subsubsection*{HMC parameters and integration of Molecular Dynamics}

The sum of the gauge and pseudofermion actions is simulated with the Hybrid
Monte Carlo (HMC) algorithm with Fourier acceleration for the U(1)
field~\cite{Batrouni:1985jn,Duane:1988vr}. The Molecular Dynamics (MD) equations
are solved by means of a symplectic multilevel integrator~\cite{Sexton:1992nu}.
We use a MD trajectory length $\tau=2$ and a three-level scheme for all our
simulations. For each level one needs to specify: the number of integration
steps, the forces to be integrated and the type of integrator.

The fermionic forces corresponding to the pseudofermion actions $S_k$ with
$k=1,\dots,N_1$ (with $N_1 \le N_2$) are integrated in the outermost level with
the Omelyan--Mryglod--Folk (OMF)~\cite{OMELYAN2003272} second-order integrator.
All other fermionic forces are integrated in the intermediate level with the OMF
fourth-order integrator. Finally, the gauge forces are integrated in the
innermost level with the OMF fourth-order integrator. The chosen values for
$N_1$ and also the number of steps for each integration level are reported in
table~\ref{tab:algo-param}.

\begin{table}
   \small
   \begin{center}
      \begin{tabular}{lccccrrr}
         \hline
         ensemble & int. steps & flavours & $\gamma$ & $[r_a, r_b]$ & \multicolumn{1}{c}{$N$ [prec.]} & \multicolumn{1}{c}{$N_1$} & \multicolumn{1}{c}{$N_2$} \\
         \hline
         \hline
         
         \texttt{A400a00b324} & (1,1,8)  & uds & 3/4 & [0.00132,8.0] & 18 [4.02e-08] & 6 & 10 \\
                              &          & c   & 1/4 & [0.25500,8.0] &  8 [7.97e-08] & 0 & 0  \\
         \hline

         \texttt{B400a00b324} & (1,1,16) & uds & 3/4 & [0.00132,8.0] & 18 [4.02e-08] & 6 & 10 \\
                              &          & c   & 1/4 & [0.25500,8.0] &  8 [7.97e-08] & 0 & 0  \\
         \hline
         \hline
         
         \texttt{A450a07b324} & (1,1,12) & u   & 1/4 & [0.0004,10.0] & 15 [4.73e-06] & 5 & 8 \\
                              &          & ds  & 1/2 & [0.0010,10.0] & 14 [5.46e-06] & 4 & 7 \\
                              &          & c   & 1/4 & [0.2000,10.0] &  7 [2.36e-06] & 0 & 0 \\ 
         \hline

         \texttt{A380a07b324} & (1,1,12) & u   & 1/4 & [0.0020,10.0] & 13 [4.00e-06] & 3 & 6 \\
                              &          & ds  & 1/2 & [0.0010,10.0] & 13 [1.38e-05] & 3 & 6 \\
                              &          & c   & 1/4 & [0.2000,10.0] &  7 [2.36e-06] & 0 & 0\\
         \hline
         \hline
         
         \texttt{A500a50b324} & (1,1,12) & u   & 1/4 & [0.00070,9.0] & 15 [2.09e-06] & 5 & 8 \\
                              &          & ds  & 1/2 & [0.00132,9.0] & 15 [1.25e-06] & 5 & 8 \\
                              &          & c   & 1/4 & [0.20000,8.0] &  9 [1.84e-08] & 0 & 0 \\
         \hline

         \texttt{A360a50b324} & (1,1,12) & u   & 1/4 & [0.0004,10.0] & 15 [4.73e-06] & 5 & 8 \\
                              &          & ds  & 1/2 & [0.0010,10.0] & 14 [5.46e-06] & 4 & 7 \\
                              &          & c   & 1/4 & [0.2000,10.0] &  7 [2.36e-06] & 0 & 0 \\
         \hline

         \texttt{C380a50b324} & (1,1,18) & u   & 1/4 & [0.0004,10.0] & 15 [4.73e-06] & 5 & 8 \\
                              &          & ds  & 1/2 & [0.0010,10.0] & 14 [5.46e-06] & 4 & 8 \\
                              &          & c   & 1/4 & [0.2000,10.0] &  7 [2.36e-06] & 0 & 0 \\

         \hline
      \end{tabular}
   \end{center}
   \caption{%
   Parameters defining the integration scheme and the pseudofermion actions. In
   the second column (\textit{int. steps}) we report the number of integration
   steps for the innermost, intermediate and outermost integration levels. For
   each degenerate multiplet of quarks specified in the third column
   (\textit{flavours}), we use a rational approximation for the operator
   $(\hat{D}^\dag \hat{D})^{-\gamma}$. We report the exponent $\gamma$, the
   chosen range $[r_a, r_b]$, the order $(N,N)$ of the rational approximation,
   and its relative uniform precision. $N_1$ denotes the number of factors
   integrated in the outermost level of the integrator, and $N_2$ denotes the
   number of factors that have been split into independent pseudofermion
   actions.
   }\label{tab:algo-param}
\end{table}

\section{Outlook}

We have presented seven (two QCD and five QCD+QED) gauge ensembles with
C$^\star$ boundary conditions. Three different values of the renormalized
fine-structure constant ($\alpha_R \simeq 0, 1/137, 0.04$) and two different
volumes ($L \simeq 1.6 \text{ fm}, 2.4 \text{ fm}$) have been considered. In all
cases we have simulated close to the SU(3)-symmetric point: in particular the
simulated up and down quarks are heavier than the physical ones, and the
simulated strange quark is lighter than the physical one. We have calculated a
number of observables: pseudoscalar meson masses, octet baryon masses, and the
$\Omega^-$ mass, the flow scale $t_0$ and the renormalized fine-structure
constant in the gradient-flow scheme. While this is only the first step of a
long-term research project, we comment here only on our goal for the near
future.

Baryon masses are unsurprisingly very noisy. Since the $\Omega^-$ is needed to
set the scale in our simulations, we plan to implement state-of-the-art
noise-reduction techniques and to invest significant resources in order to bring
the error down. In baryon correlators, we have also neglected the
quark-disconnected contractions, which are peculiar of C$^\star$ boundary
conditions. This is equivalent to calculating the masses of some
partially-quenched baryons which become degenerate with the physical baryons in
the infinite-volume limit. We plan to investigate the impact of the disconnected
contributions in detail.

The sheer number of parameters in QED+QCD simulations makes the tuning
particularly expensive. We have presented a tuning strategy, based on a number
of tricks which include mass-reweighting and linear interpolations and
extrapolations in parameters space. It turns out that on the smaller volumes,
there is no point in pushing the precision of the tuning too much, since
finite-volume effects are still significant. Volumes larger than the ones
presented here will need to be simulated in order to gain better control on
finite-volume effects. However, a first analysis already indicates that
finite-volume effects are smaller than the statistical errors on the larger
volume presented in this paper.

We also want to move towards physical quark masses, by making the up and down
quarks lighter, and the strange quark heavier. Besides the obvious
phenomenological motivation, it will be interesting to see how the cost of
simulations changes when the up quark gets lighter.

%--------------------------------------------------------------
\begin{acknowledgement}%
% !TEX root = paper.tex

We are grateful to Anian Altherr, Roman Gruber, Javad Komijani, Sofie Martins,
Paola Tavella for detailed feedback on the paper and useful conversations.
Alessandro Cotellucci and Jens L\"ucke's research is funded by the Deutsche
Forschungsgemeinschaft (DFG, German Research Foundation) - Projektnummer
417533893/GRK2575 ``Rethinking Quantum Field Theory''. The funding from the
European Union’s Horizon 2020 research and innovation program under grant
agreement No. 813942 and the financial support by SNSF (Project No.
200021\_200866) is gratefully acknowledged.
The authors gratefully acknowledge the computing time granted by the Resource
Allocation Board and provided on the supercomputer Lise and Emmy at NHR@ZIB and
NHR@Göttingen as part of the NHR infrastructure. The calculations for this
research were partly conducted with computing resources under the project
bep00085 and bep00102. The work was supported by CINECA that granted computing
resources on the Marconi supercomputer to the LQCD123 INFN theoretical
initiative under the CINECA-INFN agreement. The authors acknowledge access to
Piz Daint at the Swiss National Supercomputing Centre, Switzerland under the
ETHZ's share with the project IDs go22, go24, eth8, and s1101. The work was
supported by the Poznan Supercomputing and Networking Center (PSNC) through
grant numbers 450 and 466.
\end{acknowledgement}

\appendix
%--------------------------------------------------------------
\section{Properties of the pfaffian}\label{app:pfaffian}
% !TEX root = paper.tex

In this appendix we review the main properties of the Dirac operator and its
pfaffian. We use here $a=1$. The Dirac operator acts on the quark-antiquark
doublet
\begin{gather}
   \chi = 
   \begin{pmatrix} \psi\\
   \psi^\mathcal{C}
   \end{pmatrix} = 
   \begin{pmatrix} \psi\\
   C^{-1}\bar{\psi}^t
   \end{pmatrix}
   \ ,
\end{gather}
and can be written as a sum of terms
\begin{gather}
   D = m + D_{\mathrm{w}} + \delta D_{\mathrm{sw}}
   \ .
\end{gather}
The Wilson--Dirac operator has the standard form
\begin{gather}
   D_{\mathrm{w}} = \frac{1}{2} \sum_\mu \left\{ \gamma_\mu ( \nabla_\mu - \nabla_\mu^\dag ) - \nabla_\mu^\dag \nabla_\mu \right\}
   \ ,
\end{gather}
but the forward covariant derivative $\nabla_\mu$ is constructed keeping in mind
that the quark and antiquark fields transform under different representations of
the gauge group, one being the complex conjugate of the other, i.e.
\begin{gather}
   \nabla_\mu \chi(x)
   =
   \begin{pmatrix}
      z^{\hat{q}}(x,\mu) U(x,\mu) & 0 \\
      0 & z^{-\hat{q}}(x,\mu) U^*(x,\mu)
   \end{pmatrix}
   \chi(x + \hat{\mu}) - \chi(x)
   \ ,
\end{gather}
where $U(x,\mu)$ and $z(x,\mu)$ are the SU(3) and U(1) link variables, and
$\hat{q}$ is the electric charge of the quark in units of $q_{el}$ which appears
in eq.~\eqref{eq:U(1)_gauge_action}. With our choice $q_{el} = 1/6$, up-type
quarks have $\hat{q}=4$ while down-type quarks have $\hat{q}=-2$. In finite
volume, the definition of the forward derivative is supplemented with the
boundary conditions
\begin{gather}
   \chi(x + \tfrac{T}{a} \hat{0}) = \chi(x)
   \ , \qquad
   \chi(x + \tfrac{L}{a} \hat{k}) = K \chi(x)
   \ ,
\end{gather}
for $k=1,2,3$. The matrix
\begin{gather}
   K
   =
   \begin{pmatrix}
      0 & 1 \\
      1 & 0
   \end{pmatrix}
\end{gather}
exchanges quark and antiquark, and implements C$^\star$ boundary conditions in
this formalism. The Sheikholeslami--Wohlert term also takes into account the fact
that quark and antiquark fields transform in different representations of the
gauge group, and is given explicitly by
\begin{gather}
   \delta D_{\mathrm{sw}} =
   - \frac{1}{4} \sum_{\mu,\nu} \sigma_{\mu\nu}
   \left\{
   c_{\mathrm{sw}}^{\mathrm{SU(3)}}
   \begin{pmatrix}
      \hat{G}_{\mu\nu} & 0 \\
      0 & -\hat{G}_{\mu\nu}^*
   \end{pmatrix}
   + \hat{q} \, c_{\mathrm{sw}}^{\mathrm{U(1)}}
   \begin{pmatrix}
      \hat{F}_{\mu\nu} & 0 \\
      0 & -\hat{F}_{\mu\nu}^*
   \end{pmatrix}
   \right\}
   \ .
\end{gather}
$\hat{G}_{\mu\nu}$ and $\hat{F}_{\mu\nu}$ are the clover discretizations of the
hermitian SU(3) and U(1) field tensors, and $\sigma_{\mu\nu} = \frac{i}{2}
[\gamma_\mu,\gamma_\nu]$.

\bigskip

\begin{proposition} \label{prop:antisymmetry}
   
   The matrix $CKD$ is antisymmetric.
   
\end{proposition}

\begin{proof}
   For definiteness, we choose to work in the chiral basis for the Euclidean
   gamma matrices ($\gamma_5$ is diagonal, $\gamma_{0,2}$ are real,
   $\gamma_{1,3}$ are imaginary), and we define the charge conjugation matrix as
   \begin{gather}
      C = i \gamma_0 \gamma_2 \ .
   \end{gather}
   Notice that $C$ is imaginary and antisymmetric, and satisfies $C^2=1$. Using
   the gamma matrix identities
   \begin{gather}
      \gamma_5 C \gamma_\mu C \gamma_5 = \gamma_\mu^*
      \ , \\
      \gamma_5 C \sigma_{\mu\nu} C \gamma_5 = - \sigma_{\mu\nu}^*
      \ ,
   \end{gather}
   and the following identities involving the $K$ matrix
   \begin{gather}
      K 
      \begin{pmatrix}
         z^{\hat{q}} U & 0 \\
         0 & z^{-\hat{q}} U^*
      \end{pmatrix}
      K
      =
      \begin{pmatrix}
         z^{\hat{q}} U & 0 \\
         0 & z^{-\hat{q}} U^*
      \end{pmatrix}^*
      \ , \\
      K 
      \begin{pmatrix}
         \hat{G}_{\mu\nu} & 0 \\
         0 & -\hat{G}_{\mu\nu}^*
      \end{pmatrix}
      K
      =
      - \begin{pmatrix}
         \hat{G}_{\mu\nu} & 0 \\
         0 & -\hat{G}_{\mu\nu}^*
      \end{pmatrix}^*
      \ , \\
      K 
      \begin{pmatrix}
         \hat{F}_{\mu\nu} & 0 \\
         0 & -\hat{F}_{\mu\nu}^*
      \end{pmatrix}
      K
      =
      - \begin{pmatrix}
         \hat{F}_{\mu\nu} & 0 \\
         0 & -\hat{F}_{\mu\nu}^*
      \end{pmatrix}^*
      \ ,
   \end{gather}
   one easily proves
   \begin{gather}
      \gamma_5 C K D K C \gamma_5 = D^* \ ,
   \end{gather}
   or, equivalently,
   \begin{gather}
      C K D = \gamma_5 D^* \gamma_5 C K = D^t C K = - ( C K D )^t \ .
   \end{gather}
   In the third equality we have used $\gamma_5$-hermiticity of the Dirac
   operator, and in the last step we have used $C^t = -C$ and $K^t = K$.
   
\end{proof}

\bigskip

\begin{proposition} \label{prop:degeneracy}
   
   The spectrum of the operator $Q = \gamma_5 D$ is doubly degenerate.
   
\end{proposition}

\begin{proof}
   
   The matrix $U = C K \gamma_5$ has the following properties
   \begin{gather}
      U = U^\dag = U^{-1} = - U^t = U^*
      \ .
   \end{gather}
   Proposition~\ref{prop:antisymmetry} implies
   \begin{gather}
      U Q^* = C K D^* = - (C K D)^* = (C K D)^\dag = \gamma_5 D \gamma_5 C K = Q U \ .
   \end{gather}
   As a simple application of the above relations, it follows that:
   \begin{enumerate}
      
      \item if $v$ is an eigenvector of $Q$ with eigenvalue $\lambda$, then $U v^*$ is also an eigenvector of $Q$ with the same eigenvalue:
      \begin{gather}
         Q U v^* = U Q^* v^* = U (\lambda v)^* = \lambda U v^* \ ,
      \end{gather}
      where the last equality follows from the fact that $\lambda$ is real;
      
      \item $v$ and $U v^*$ are orthogonal:
      \begin{gather}
         ( U v^* , v ) = v^t U^\dag v = v^t U v = 0 \ ,
      \end{gather}
      where the last equality follows from the fact that $U$ is antisymmetric;
      
      \item if $w$ is orthogonal to $v$ and $U v^*$, then $U w^*$ is also orthogonal to $v$ and $U v^*$:
      \begin{gather}
         ( Uw^* , v ) = w^t U v = - v^t U w = - ( v^*, Uw ) = - ( Uv^*, w ) = 0
         \ , \\
         ( U w^* , U v^* ) = w^t U^2 v^* = v^\dag w = (v,w) = 0
         \ .
      \end{gather}
      
   \end{enumerate}
   A straightforward modification of the Gram--Schmidt algorithm, in which one
   alternates an orthogonalization step with the construction of an eigenvector of
   the form $U v_i^*$, allows to prove that the degeneracy $d$ of any
   eigenvector is even, and an orthonormal basis for the corresponding eigenspace
   can be chosen of the form:
   \begin{gather}
      v_1, \ U v_1^*,\ v_2, \ U v_2^*, \ \dots \ , \ v_{d/2}, \ U v_{d/2}^* \ .
   \end{gather}
   
\end{proof}

\bigskip

\begin{proposition} \label{prop:pfaffian}
   
   Let $\lambda_{n=1,\dots,12V} \in \mathbb{R}$ be the list of eigenvalues of
   $Q$, each of them appearing a number of times equal to half their degeneracy.
   Then the following formula holds
   \begin{gather}
       \pf \, (CKD)
       =
       \prod_{n=1}^{12V} \lambda_n
       \ . \label{eq:app:pf-lambda}
   \end{gather}
   
\end{proposition}

\begin{proof}
   
   The pfaffian of a matrix is a polynomial in the entries of the matrix. Since
   $D$ depends linearly on the bare mass $m$, then the function
   \begin{gather}
      f(m) = \pf \, [CKD(m)]
   \end{gather}
   is a polynomial in $m$. Since $Q$ is hermitian and depends linearly on $m$,
   for every $m$ its eigenvalues $\lambda_n$ can be labeled in such a way that
   they are analytic functions of $m$ even at level
   crossings~\cite{Kato:101545}. Then the function
   \begin{gather}
      g(m) = \prod_{n=1}^{12V} \lambda_n(m)
   \end{gather}
   is analytical for every real value of $m$.
   
   For $m \to +\infty$ one has $D \simeq m I_{24V}$, which implies
   \begin{gather}
      \lim_{m \to +\infty} f(m) = \lim_{m \to +\infty} m^{12V} \pf( C \otimes K \otimes I_{3V} ) = \lim_{m \to +\infty} m^{12V} = +\infty
      \ .
   \end{gather}
   For $m \to +\infty$ one has $Q \simeq m I_{6V} \otimes \gamma_5$, which
   implies that half of the eigenvalues are asymptotically equal to $+m$ and the
   other half are asymptotically equal to $-m$. Therefore
   \begin{gather}
      \lim_{m \to +\infty} g(m) = \lim_{m \to +\infty} \prod_{n=1}^{12V} \lambda_n = \lim_{m \to +\infty} m^{6V} (-m)^{6V} = +\infty
      \ .
   \end{gather}
   From the two limits it follows that a value $M$ exists such that both $f(m)$
   and $g(m)$ are positive for every $m > M$.
   
   Using properties of the pfaffian and the determinant one easily shows that,
   for every $m$,
   \begin{gather}
      f(m)^2
      =
      \pf \, [CKD(m)]^2
      =
      \det \left[ D(m) \right]
      =
      \det \left[ Q(m) \right]
      =
      \prod_{n=1}^{12V} \lambda_n^2(m)
      =
      g(m)^2
      \ .
   \end{gather}
   For $m>M$ both functions are positive, therefore the above equality implies
   $f(m) = g(m)$.
   
   Since $f(m)$ and $g(m)$ are analytic function of $m$, and they are equal for
   every $m > M$, it follows that they are equal everywhere.      
   
\end{proof}

\clearpage
\small
\addcontentsline{toc}{section}{References}
\bibliographystyle{JHEP}
\bibliography{non-inspire,inspire}

\end{document}